\documentclass[10pt,conference,letterpaper]{IEEEtran}
\IEEEoverridecommandlockouts
\usepackage{cite,balance}
\usepackage{amsmath,amssymb,amsfonts,amsthm}
\def\BibTeX{{\rm B\kern-.05em{\sc i\kern-.025em b}\kern-.08em
    T\kern-.1667em\lower.7ex\hbox{E}\kern-.125emX}}

\newtheorem{theorem}{Theorem}
\newtheorem{lemma}{Lemma}
\newtheorem{proposition}{Proposition}

\usepackage[utf8]{inputenc}
\usepackage[T1]{fontenc}

\usepackage[binary-units]{siunitx}
\DeclareSIUnit\Byte{\text{Byte}}
\usepackage{etoolbox}
\robustify\bfseries

\usepackage{algorithm}
\usepackage[noend]{algpseudocode}

\usepackage{booktabs}
\usepackage[dvipsnames]{xcolor} %
\usepackage{pgfplotstable}
\usepackage{colortbl}
\usepackage{tipa}

\usepackage{tikz}
\usetikzlibrary{shapes.arrows,chains,calc}
\usetikzlibrary{decorations.pathreplacing,matrix}

\usepackage[labelformat=parens]{subcaption}

\makeatletter
\newcommand{\linebreakand}{%
\end{@IEEEauthorhalign}%
\hfill%
\mbox{}%
\par%
\relax%
\mbox{}%
\hfill%
\begin{@IEEEauthorhalign}%
}
\makeatother
\newcommand{\Problem}[1]{\textsc{#1}}
\newcommand{\Kilo}{\ensuremath{\mathrm{K}}}
\newcommand{\Mega}{\ensuremath{\mathrm{M}}}

\newcommand{\TabLabel}[1]{\label{tab:#1}}
\newcommand{\Table}[1]{Table~\ref{tab:#1}}

\newcommand{\FigLabel}[1]{\label{fig:#1}}
\newcommand{\Figure}[1]{Fig.~\ref{fig:#1}}
\newcommand{\Figures}[2]{Figs.~\ref{fig:#1} and~\ref{fig:#2}}
\newcommand{\FiguresSub}[2]{Figs.~\ref{fig:#1},\,\subref{fig:#2}}

\newcommand{\SectLabel}[1]{\label{sect:#1}}
\newcommand{\Section}[1]{Sect.~\ref{sect:#1}}

\newcommand{\ThmLabel}[1]{\label{thm:#1}}
\newcommand{\Theorem}[1]{Theorem~\ref{thm:#1}}
\newcommand{\LemLabel}[1]{\label{lem:#1}}
\newcommand{\Lemma}[1]{Lemma~\ref{lem:#1}}
\newcommand{\PropLabel}[1]{\label{prop:#1}}

\newcommand{\ie}{i.e.}
\newcommand{\eg}{e.g.}
\newcommand{\etal}{et~al.}
\newcommand{\wrt}{w.r.t.}

\usepackage{wrapfig}
\usepackage{calc}
\newcommand{\erclogowrapped}[1]{%
\setlength\intextsep{0pt}%
\begin{wrapfigure}[3]{r}{#1*\real{1.1}}%
\includegraphics[width=#1]{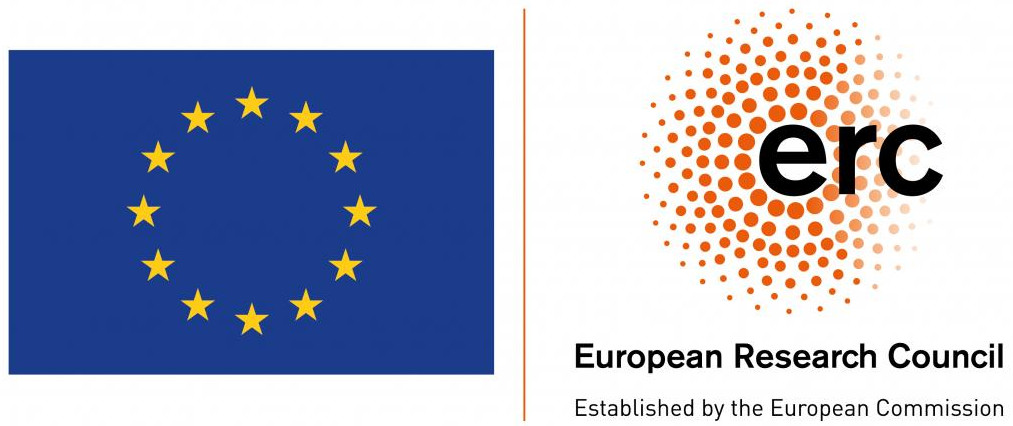}%
\end{wrapfigure}%
}

\newcommand{\N}{\ensuremath{\mathbb{N}}}
\newcommand{\Set}[1]{\ensuremath{\{#1\}}}
\newcommand{\Card}[1]{\ensuremath{\left| #1 \right|}}
\newcommand{\Range}[2]{\ensuremath{\left[ #1, #2 \right]}}

\newcommand{\Edge}[2]{\ensuremath{\{#1,#2\}}}
\newcommand{\Neigh}{\ensuremath{N}}
\newcommand{\Degree}{\ensuremath{\mathrm{deg}}}

\newcommand{\MaxDegree}{\ensuremath{\Delta}}

\newcommand{\Path}{\ensuremath{\mathcal{P}}}
\newcommand{\Matching}{\ensuremath{\mathcal{M}}}

\newcommand{\MatrixMult}{\ensuremath{\omega}}

\newcommand{\NP}{\ensuremath{\mathcal{NP}}}
\newcommand{\bigO}{\ensuremath{\mathcal{O}}}
\newcommand{\softO}{\ensuremath{\tilde{\mathcal{O}}}}
\newcommand{\eps}{\epsilon}

\DeclareMathOperator{\Sort}{sort}

\newcommand{\OPT}{\textsf{OPT}}
\newcommand{\BEST}{\textsf{BEST}}
\usepackage{nicefrac}

\newcommand{\halfapprox}{$\frac{1}{2}$-approximation}
\newcommand{\halfalpha}{\ensuremath{\nicefrac{1}{2}}}
\newcommand{\twoaug}{2-augmentation}

\newcommand{\Demand}{\ensuremath{\mathcal{D}}}

\newcommand{\B}{\ensuremath{k}}

\newcommand{\MDM}{\textsc{DjM}}
\newcommand{\MDMlong}{\Problem{$\B$-Disjoint Matching}}
\newcommand{\MaxWeight}{\ensuremath{D}}

\newcommand{\DjMatching}{\ensuremath{{\mathcal{M}}}}
\newcommand{\MatchAlg}{\ensuremath{\mathcal{A}_\mathrm{M}}}

\newcommand{\AlgName}[1]{\texttt{#1}}
\newcommand{\ROMA}{\AlgName{ROMA}}
\newcommand{\Swaps}{\AlgName{Swaps}}
\newcommand{\LocalSwaps}{\AlgName{LocalSwaps}}
\newcommand{\GlobalSwaps}{\AlgName{GlobalSwaps}}
\newcommand{\Blossom}{\AlgName{Blossom}}
\newcommand{\BlossomIt}{\AlgName{Blossom-It}}

\newcommand{\GPA}{\AlgName{GPA}}
\newcommand{\GPAIt}{\AlgName{GPA-It}}
\newcommand{\lROMA}{l_{\ROMA}}

\newcommand{\GreedyIt}{\AlgName{Greedy-It}}

\newcommand{\BGreedyExt}{\AlgName{bGreedy\&Extend}}

\newcommand{\misragries}{\AlgName{MG}}
\newcommand{\nodecentered}{\AlgName{NodeCentered}}
\newcommand{\NCthresh}{\ensuremath{\theta}}
\newcommand{\NCmax}{\texttt{MAX}}
\newcommand{\NCavg}{\texttt{AVG}}
\newcommand{\NCmed}{\texttt{MEDIAN}}
\newcommand{\NCsum}{\texttt{SUM}}
\newcommand{\NCbsum}{\texttt{kSUM}}

\newcommand{\kEC}{\AlgName{k-EC}}
\newcommand{\kECcc}{\texttt{CC}}
\newcommand{\kEClc}{\texttt{LC}}
\newcommand{\kECrl}{\texttt{RL}}
\newcommand{\kEClf}{\texttt{LF}}

\newcommand{\Instance}[1]{\texttt{#1}}
\usepackage{hyperref}
\hypersetup{
colorlinks,
}
\usepackage{orcidlink}
\newcommand*{\subparagraph}[1]{\paragraph{#1}}

\title{Fast and Heavy Disjoint Weighted Matchings\\for
Demand-Aware Datacenter Topologies%
\thanks{
\erclogowrapped{5\baselineskip}
This project has received funding from the
European Research Council (ERC) under the European Union's Horizon 2020
research and innovation programme (Grant agreement No.\ 101019564
``The Design of Modern Fully Dynamic Data Structures (MoDynStruct)''
and
Grant agreement No.\ 864228
``Self-Adjusting Networks (AdjustNet)'')%
,
and from the
Austrian Science Fund (FWF) project ``Fast Algorithms for a Reactive Network
Layer (ReactNet)'', P~33775-N, with additional funding from the \textit{netidee SCIENCE
Stiftung}, 2020--2024.
}}

\begin{document}
\author{%
\thanks{%
\{kathrin.hanauer,\,monika.henzinger,\,jonathan.trummer\}@univie.ac.at,
stefan.schmid@tu-berlin.de

© 2022 IEEE. Personal use of this material is permitted. Permission from IEEE must be 
obtained for all other uses, in any current or future media, including 
reprinting/republishing this material for advertising or promotional purposes, creating new 
collective works, for resale or redistribution to servers or lists, or reuse of any copyrighted 
component of this work in other works.
}
\IEEEauthorblockN{%
Kathrin Hanauer
\orcidlink{0000-0002-5945-837X} 
}
\IEEEauthorblockA{\textit{Faculty of Computer Science}, \textit{University of Vienna},
Austria%
}
\and
\IEEEauthorblockN{%
Monika Henzinger
\orcidlink{0000-0002-5008-6530} 
}
\IEEEauthorblockA{\textit{Faculty of Computer Science},
\textit{University of Vienna},
Austria%
}%
\linebreakand%
\IEEEauthorblockN{%
Stefan Schmid
\orcidlink{0000-0002-7798-1711} 
}
\IEEEauthorblockA{\textit{Faculty of Computer Science},
\textit{University of Vienna}, Austria\\
and \textit{TU Berlin}, Berlin, Germany%
}
\and
\IEEEauthorblockN{%
Jonathan Trummer
\orcidlink{0000-0002-1086-4756} 
}
\IEEEauthorblockA{\textit{Faculty of Computer Science}, \textit{University of Vienna},
Austria%
}%
}
\maketitle
\begin{abstract}
Reconfigurable optical topologies promise
to improve the performance in datacenters by
dynamically optimizing the physical network in a demand-aware manner.
State-of-the-art optical technologies allow to establish and
update direct connectivity (in the form of \emph{edge-disjoint matchings})
between top-of-rack switches within
microseconds or less.
However, to fully exploit temporal structure in the demand,
such fine-grained reconfigurations also require fast algorithms
for optimizing the interconnecting matchings.

Motivated by the desire to offload a maximum amount of demand to the
reconfigurable network, this paper initiates the study of fast algorithms
to find $\B$ disjoint heavy matchings in graphs. We present and analyze 
six algorithms, based on iterative matchings, b-matching, edge coloring, and node-rankings.
We show that the problem is generally \NP{}-hard and
study the achievable approximation ratios.

An extensive empirical evaluation
of our algorithms on both real-world and synthetic traces (\num{88} in total),
including traces collected in
Facebook datacenters and in HPC clusters
reveals that all our algorithms provide high-quality matchings,
and also very fast ones come within \SI{95}{\percent} or more of the best solution. However, the
running times differ significantly and what is the best algorithm
depends on $\B$ and the acceptable runtime-quality tradeoff.
\end{abstract}
\section{Introduction}\SectLabel{intro}

With the popularity of data-centric applications, network traffic in datacenters is growing explosively~\cite{singh2015jupiter,Roy2015InsideTS}. 
Accordingly, over the last years, several novel datacenter topologies have been proposed to improve network efficiency and performance, e.g.,~\cite{valadarsky2016xpander,guo2009bcube,singla2012jellyfish}. 
These network topologies typically have in common that they are oblivious to the traffic they serve.

Emerging reconfigurable optical technologies enable an intriguing alternative to existing datacenter network designs~\cite{osn21,zhang2021gemini,sigmetrics22cerberus}: these technologies allow to enhance existing datacenter networks with reconfigurable optical matchings, \eg, one disjoint matching per optical circuit switch~\cite{projector,mellette2017rotornet,helios,wang2010c,osa,ballani2020sirius,ancs18}. 
These matchings can be adapted towards the traffic demand, exploiting temporal and spatial structure~\cite{sigmetrics20complexity}. 
State-of-the-art technologies allow in principle to change such matchings within microseconds or even less~\cite{ballani2020sirius}. Given these reconfiguration times, the bottleneck becomes now how to \emph{compute} such matchings fast in the control plane.

Accordingly, this paper initiates the study of fast algorithms to find $\B$ disjoint weighted matchings in graphs. Here, $\B$ is the number of optical circuit switches, each of which provides one reconfigurable matching. 
The matchings should be heavy, \ie, carry a maximal amount of traffic.
Existing algorithmic work on the matching problem typically focuses on computing a single matching. 

More formally, this paper considers the following model and terminology. We
are given a weighted graph $G=(V,E,\Demand)$ where $\Demand: V \times V \rightarrow \N_0$ describes
a \emph{demand} for each pair of vertices.
A \emph{matching} is a subset of edges that have no common vertices.
The goal is to find $\B$ pairwise edge-disjoint matchings $\Matching_1, \dots,
\Matching_{\B}$, such that $\sum_{1 \leq i \leq \B} \Demand(\Matching_i)$ is
maximized.
We show that this problem is \NP{}-hard.

Note that this problem is different from the well-known $b$-matching problem: Given a triangle, a 2-matching could choose all edges in the graph, while 2 disjoint matchings consist of 2 edges of the graph.

Perhaps a natural approach to find such maximum weight $\B$ disjoint matchings would be to repeatedly
remove the edges of individual maximum weight matchings from the graph, leading to a polynomial-time algorithm.
However, this approach is not optimal:
Consider a triangle where each vertex is additionally adjacent to a further vertex of degree one.
Thus, the graph has six vertices and six edges.
Now assume unit edge weights and consider $\B=3$.
The unique maximum weight matching consists of exactly the three edges incident to the degree-one vertices.
By removing them we are left with a triangle, and hence the second maximum weight matching has size one, as does the third.
An optimal solution, however, has size six: Each of the three matchings consists of one triangle edge plus the unique edge
incident to the remaining triangle vertex.

\paragraph*{Contributions}
Motivated by novel optical technologies which allow to enhance fixed datacenter topologies with reconfigurable matchings, this paper studies algorithms for the fast computation of $\B$ heavy disjoint matchings.
We show that the problem is \NP{}-hard and propose six efficient algorithms:
three algorithms are based on the iterative computation of simple matchings,
one algorithm leverages a connection to the related $b$-matching problem,
one algorithm is an adaptation of an edge-coloring algorithm,
and one node-centered algorithm
uses a rating function that depends on the weights of a node's incident edges. Additionally, we
study three postprocessing strategies to improve the weight of an existing matching,
and discuss the achievable approximation ratios.

We perform an extensive evaluation of the quality and running time
of our algorithms on a diverse set of instances, which include both
real-world traces as well as synthetic traces, \num{88} instances in total.
In particular, we consider six traces from Facebook datacenters,
four traces from a High-Performance Computing (HPC) cluster, three widely-used synthetic pFabric traces,
nine instances from the Florida sparse matrix collection, and 66 Kronecker graphs.

Our empirical results show that our algorithms constantly compute high-quality matchings.
The running times however vary and which is the best algorithm depends on the value of $\B$
and the affordable tradeoff between running time and solution quality.
For small values of $\B$, an iterative approach is most attractive, especially when combined
with a local swapping strategy: the algorithm \GPAIt{},
which is based on the \AlgName{Global Paths} matching algorithm,
combined with the \LocalSwaps{} postprocessing routine provides low running times
and high-quality matchings whose weight is within \SI{95}{\percent} or more of the best algorithm (executed with a \SI{4}{\hour} time limit).
For larger values of $\B$, our edge coloring-based algorithm \kEC{} provides the best performance
as its running time barely increases with $\B$; for $\B \geq 4$, its quality
score is always at least within \SI{93}{\percent} of the best algorithm, and at
least \SI{96}{\percent} on average.
If running time is not of concern, the iterative algorithm \BlossomIt{} can be an attractive choice.

As a contribution to the research community,
to ensure reproducibility and to facilitate follow-up work,
we will make all our experimental artefacts including our implementation (as open source) publicly available together with this paper.

\paragraph*{Organization}
The remainder of this paper is organized as follows.
\Section{preliminaries} introduces preliminaries and discusses related problems and prior work.
We present our algorithms in \Section{algorithms} and report on our experimental results in \Section{experiments}.
We conclude by discussing future research directions in \Section{conclusion}.
\section{Preliminaries}\SectLabel{preliminaries}
\subsection{Basic Definitions}
We model our problem as a simple, undirected, edge-weighted graph
$G=(V,E,\Demand)$ with vertex set $V$, edge set $E$, and a non-negative integer
\emph{weight} or \emph{demand} $\Demand: V \times V \rightarrow \N_0$ for each pair of vertices,
which corresponds to the amount of communication (data flow) between them.
We assume the demand to be symmetric, \ie, $\Demand(u, v) = \Demand(v, u)$ for
all $u \ne v \in V$ and $\Demand(u, v) > 0$ iff $\Edge{u}{v} \in E$.
As shorthand notation, we also write $\Demand(e) = \Demand(u, v)$ for $e =
\Edge{u}{v}$.
Note that because $G$ is simple, $\Demand(v, v) = 0$ always.
As usual, $n = \Card{V}$ and $m = \Card{E}$.
Furthermore, we denote by $D = \max_{u,v \in V} \Demand(u, v)$ the maximum
demand.
An edge $e = \Edge{u}{v}$ is \emph{incident} to its \emph{end vertices} $u$ and
$v$, and $u$ and $v$ are said to be \emph{adjacent}.
The \emph{neighborhood} of a vertex $v$ is $\Neigh(v) = \Set{ u \mid
\Edge{u}{v} \in E}$ and its \emph{degree} is $\Degree(v) = \Card{\Neigh(v)}$.
We denote the \emph{maximum degree} by $\MaxDegree = \max_{v\in V} \Degree(v)$.
A \emph{path} $\Path$ is a sequence of edges $\Path = (e_1,e_2,\dots,e_k)$ such
that $e_i$ and $e_{i+1}$, $1 \leq i < k$, share a common end vertex and no
vertex appears more than once.

A \emph{matching} $\Matching \subseteq E$ is a set of edges such that no vertex is incident
to two edges contained in $\Matching$.
In our context,
the \emph{weight} of a matching $\Matching$ is
$\Demand(\Matching) = \sum_{e \in \Matching} \Demand(e)$.
An edge $e$ is said to be \emph{matching} if $e \in \Matching$ and
\emph{non-matching} otherwise.
A vertex is said to be \emph{free} (\wrt{}~$\Matching$) if it is not incident
to a matching edge.
An \emph{alternating path} is a path that alternatingly consists of matching
and non-matching edges.
Given a matching $\Matching$ and a path $\Path$, we denote by $\Matching \oplus
\Path$ the set of edges obtained as the symmetric difference of $\Matching$ and
the edges in $\Path$.
An \emph{augmenting path} is an alternating path $\Path$ such that $\Matching
\oplus \Path$ is a matching and $\Demand(\Matching) < \Demand(\Matching \oplus
\Path)$.
A \emph{$\B$-disjoint matching} $\DjMatching$ is a collection of $\B$ matchings
$(\Matching_1, \dots, \Matching_{\B})$ that are pairwise disjoint.
We slightly abuse notation and use $\DjMatching$ both for the collection of
disjoint matchings as well as their union $\Matching_1 \cup \dots \cup
\Matching_{\B}$.

In this paper, we study the \emph{\MDMlong{} ($\B$-\MDM)} problem:
Given a graph $G = (V, E, \Demand)$ and an integer $\B \in \N$, find a
$\B$-disjoint matching $\DjMatching = (\Matching_1, \dots, \Matching_{\B})$
such that
$\Demand(\DjMatching) = \sum_{1 \leq i \leq \B} \Demand(\Matching_i)$ is maximized.
If the maximum demand $D=1$, we speak of the \emph{unweighted} $k$-disjoint matching problem.
\subsection{Related Problems and Prior Work}\SectLabel{related}
Our problem is especially motivated by reconfigurable datacenter networks in which optical switches can be used to augment a given fixed (electrical) topology, typically a Clos topology~\cite{singh2015jupiter}, with additional matchings between the top-of-rack switches~\cite{helios,wang2010c,osa,projector,amir2021optimal,sigmetrics22cerberus,ballani2020sirius}. In prior work, these matchings are typically optimized individually and not for runtime~\cite{sigact19}. 
In this regard, our paper is also related to graph augmentation problems~\cite{gozzard2018converting,meyerson2009minimizing} where a given (fixed) graph needs to be enhanced with an optimal number of ‘‘extra edges’’, sometimes also referred to as ‘‘ghost edges’’~\cite{papagelis2011suggesting}: the objective in this literature is typically to provide small world properties~\cite{parotsidis2015selecting} or minimize the network diameter~\cite{bilo2012improved,demaine2010minimizing}. However, these algorithms are not directly applicable to our problem where we need to add entire matchings rather than individual edges.

The \Problem{Weighted Matching} problem essentially corresponds to the special
case of $1$-\MDM{}. %
Edmonds~\cite{edmonds_1965_cardinality,edmonds1965maximum} was the first to
give a polynomial-time algorithm, which is known as the
\emph{blossom algorithm} and has a running time of $\bigO(mn^2)$.
In a series of improvements~\cite{%
gabow1974,
DBLP:journals/jacm/Gabow76,%
lawler1976,%
DBLP:journals/siamcomp/GalilMG86,%
DBLP:journals/jcss/Gabow85,%
DBLP:journals/jacm/GabowGS89}, the running time
has been reduced further to $\bigO(n(m + n\log
n))$~\cite{DBLP:conf/soda/Gabow90} for the general case and to
$\bigO(m\sqrt{n\log n}\log(nD))$~\cite{DBLP:journals/jacm/GabowT91} for integer
weights. %
For dense graphs, the fastest algorithm on integer-weighted graphs is
randomized and runs in time
$\softO(Dn^\MatrixMult)$~\cite{DBLP:journals/jacm/CyganGS15}, where
$\MatrixMult$ is the exponent in the running time of fast matrix multiplication
algorithms ($\MatrixMult < 2.38$~\cite{LeGall14}).
A widely known simple greedy algorithm~(cf.~\Section{algorithms}) yields
a $\frac{1}{2}$-approximation and runs in time $\bigO(m \log n)$.
Both approximation ratio and running time have been subject to improvement over
the years%
, leading to a $(1-\eps)$-approximation algorithm with $\bigO(m/\eps
\log(1/\eps))$ running time for arbitrary edge weights and $\bigO(m/\eps \log
D)$ running time for integer weights%
~\cite{DBLP:journals/jacm/DuanP14}.

Different algorithms have been proposed and evaluated to tackle the problem in
practice.
Drake and Hougardy~\cite{DBLP:conf/wea/DrakeH03} experimentally compared the
already mentioned greedy algorithm to the \AlgName{LD} algorithm by
Preis~\cite{DBLP:conf/stacs/Preis99} and to the Path Growing Algorithm
(\AlgName{PGA})~\cite{DBLP:journals/ipl/DrakeH03} and showed that
a variant of \AlgName{PGA}, \AlgName{PGA'}, performs very well in practice.
Maue and Sanders~\cite{DBLP:conf/wea/MaueS07} later suggested the Global Paths
Algorithm (\GPA{}) and showed in an extensive study that in combination with a
postprocessing routine called \ROMA{}, it yields the best experimental results in
comparison to the simple greedy algorithm and \AlgName{PGA'}.

\Problem{(Weighted) Perfect Matching} is a restricted version that disallows free
vertices.
It can be solved in polynomial time by a variant of the blossom
algorithm~\cite{kolmogorov2009blossom}.

\Problem{(Weighted) $b$-Matching} is a generalization of \Problem{Weighted
Matching}, where each vertex may be incident to up to $b$ edges contained in
the matching.
In contrast to \MDMlong{}, a $b$-matching need not be composed of
$b$ pairwise disjoint $1$-matchings.
Thus, every $\B$-disjoint matching is a $\B$-matching, but not necessarily
vice-versa (the edges of a triangle, \eg, form a 2-matching but not a
2-disjoint matching).
This problem can be solved exactly in $\bigO(\min\{bn, n\log n\}(m+n\log
n))$ time~\cite{DBLP:journals/talg/Gabow18} and approximated by a greedy algorithm analogous to
the greedy weighted matching algorithm with a performance guarantee of
$\frac{1}{2}$~\cite{DBLP:conf/esa/Mestre06} (cf.~\Section{algorithms}).
A $(1+\eps)$-approximation can be achieved in $\bigO{}_\eps(m\alpha{}(m,n))$
time~\cite{DBLP:journals/corr/HuangP17b}.

Khan \etal{}~\cite{DBLP:journals/siamsc/KhanPPSSMHD16} compared the performance
of the $b$-matching variants of the simple greedy algorithm, \AlgName{PGA},
\AlgName{PGA'}, and \AlgName{LD} to their new algorithm \AlgName{b-SUITOR},
which computes the same solution as the greedy algorithm, but is parallelizable
and faster than \AlgName{PGA'}.
Recently, algorithms for the $b$-matching problem have been evaluated in the
online setting in a similar context of data center
reconfiguration~\cite{DBLP:journals/sigmetrics/BienkowskiFMS20}.

The \Problem{Edge Coloring} problem consists in determining the
\emph{chromatic index} of a graph, \ie, the minimum number of colors required
to assign each edge a color such that edges incident to a common vertex receive
different colors.
The \MDMlong{} problem is hence equivalent to finding a maximum-weight subgraph
with chromatic index $\B$.
Whereas $\MaxDegree$ naturally gives a lower bound on the chromatic index, an
upper bound is given by $\MaxDegree + 1$~\cite{vizing1964}.
In general, it is $\NP$-hard to decide whether a graph has chromatic index
$\MaxDegree$ or $\MaxDegree + 1$~\cite{DBLP:journals/siamcomp/Holyer81a}
already if the graph is cubic, i.e., $\Degree(v) = 3$ for all $v \in V$:
\begin{theorem}[Holyer~\cite{DBLP:journals/siamcomp/Holyer81a}]\ThmLabel{ec-np}
It is $\NP$-complete to determine whether the chromatic index of a cubic graph
is $3$ or $4$.
\end{theorem}
\begin{proposition}\PropLabel{npc}
The \MDMlong{} problem is $\NP$-hard already in the unweighted
case and for $\B = 3$.
\end{proposition}
\begin{proof}%
Consider a cubic graph $G$.
There is a one-to-one correspondence between deciding the chromatic index of $G$ and deciding whether $\B = 3$ disjoint matchings of weight
at least $m$ exist in a graph where $D = 1$:
The set of edges of the same color in the 3-coloring give three matchings
and three disjoint matchings in $G$ give a 3-coloring of $G$.
\end{proof}
Recall that for $\B =1$, \MDMlong{} is equal to \Problem{Weighted
Matching} and thus solvable in polynomial time, whereas for $\B =2$ the
complexity is still unknown.
Computing a maximum weight matching, removing it, and computing a second does
not necessarily yield an optimal solution for $\B=2$, as we show in
\Section{algorithms}.
Note that it is easy to tell whether a graph with maximum degree two can be
colored with two or three colors:
Two colors always suffice unless the graph has an odd-length cycle.
We give below various polynomial-time algorithms whose running time is
polynomial in $n$, $m$, and $\B$ (observe that $\B \leq \MaxDegree + 1$).

In consequence of \Theorem{ec-np}, \Problem{Edge Coloring} is inapproximable
within a factor less than $4/3$.
This implies the following inapproximability result for \MDM{}:
\begin{proposition}\PropLabel{approx}
It is $\NP$-hard to approximate the \MDMlong{} problem
within a factor of $(1-\eps)$ for any $1/m \ge \eps > 0$.
\end{proposition}
\begin{proof}
Consider a cubic graph $G$ with $\MaxWeight=1$.
If $G$'s chromatic index is $3$, then every algorithm for $\B$-\MDM{} with an
approximation ratio strictly greater than
$(1-\frac{1}{m})$ must compute a solution for $\B = 3$ that contains more than
$m(1-\frac{1}{m}) = m-1$ edges, which implies it must contain $m$ edges. But if it is $4$, it can contain at most $m-1$ edges. Thus any such algorithm can be used to decide whether the chromatic index is 3 or 4.
\end{proof}
\Problem{Edge Coloring} can be solved to optimality in $\bigO(2.2461^m)$ time
by first obtaining the line graph and then applying an algorithm for
\Problem{Vertex
Coloring}~\cite{DBLP:journals/siamcomp/BjorklundHK09}.
Misra and Gries~\cite{DBLP:journals/ipl/MisraG92} gave an algorithm that
constructs a $\Delta+1$ coloring in time $\bigO(nm)$, whereas an algorithm that
greedily colors the edges with the first available color can use up to
$2\Delta-1$ colors, which is optimal in the online
setting~\cite{DBLP:journals/ipl/Bar-NoyMN92}. %

The \Problem{$r$-factorization} problem asks for a partition of a graph's edge set into a
disjoint collection of $r$-regular spanning subgraphs, called $r$-factors.
Thus, a $1$-factor is a perfect matching, and only regular graphs with an even
number of vertices can have a $1$-factorization, which is then equivalent to
the \Problem{Edge Coloring} problem.

\section{Algorithms}\SectLabel{algorithms}
In the following, we propose and engineer different approaches to obtain
disjoint matchings.
As the problem is \NP{}-hard and the field of application requires solutions
computable within fractions of a second, we concentrate on algorithms from
which we expect good, though not necessarily optimal quality.
Our approaches are inspired by and partially also built on methods for the
related problems of \Problem{Weighted Matching}, \Problem{Weighted $b$-Matching},
and \Problem{Edge Coloring},
and are evaluated on a diverse set of instances in \Section{experiments}.

\Table{complexities} provides an overview of all algorithms and lists
their time complexity as well as approximation guarantees.
\begin{table}[tb]
\centering
\caption{Running Time complexities of various algorithms.}\TabLabel{complexities}
\begin{tabular}{@{}lcc@{}}
\toprule
Algorithm
& Running Time Complexity
& Approximation
\\
\midrule
\BlossomIt{}%
& $\bigO(\B n(m + n \log n))$
& $\leq \nicefrac{7}{9}$%
\\
\GreedyIt{}%
& $\bigO(\Sort(m) + \B m)$
& \halfalpha{} %
\\
\GPAIt{}%
& $\bigO(\Sort(m) + \B m)$
& $\leq\halfalpha{}$\\
\BGreedyExt{}%
& $\bigO(\Sort(m) + \B n^2)$
& $\leq\halfalpha{}$\\
\nodecentered{}%
& $\bigO(\Sort(n) + n \cdot \Sort(\MaxDegree{}) + \B m)$
& $\leq\halfalpha{}$\\
\kEC{}%
& $\bigO(\Sort(m) + \B n^2)$
& $\leq\halfalpha{}$\\
\bottomrule
\end{tabular}
\end{table}
\subsection{Algorithms based on Weighted Matching (\AlgName{*-It})}
Given an algorithm $\MatchAlg$ for the weighted matching problem, a
straightforward approach to obtain $\B$ disjoint matchings of
large total weight consists in running $\MatchAlg$ $\B$ times and making the
set of matching edges ``unavailable'' to subsequent runs of $\MatchAlg$.
As the disjoint matchings are obtained iteratively, we use the suffix
\textsc{-It} for algorithms following this scheme.
We study different options for $\MatchAlg$:

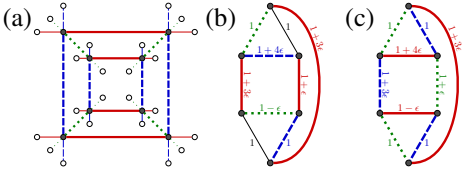
\begin{figure}[tb]
\centering
\begin{tikzpicture}
\node[inner sep=0pt] (hyper) {\includegraphics[height=2.2cm]{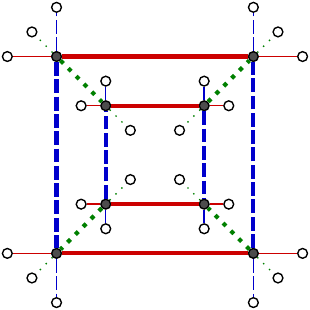}\phantomsubcaption\label{fig:hypercube}};
\node[inner sep=0pt,anchor=west,xshift=.5cm] (blossom-it-k2) at (hyper.east) {\includegraphics[height=2.2cm]{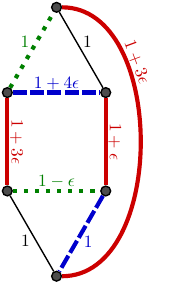}\phantomsubcaption\label{fig:bit-k2}};
\node[inner sep=0pt,anchor=west,xshift=.5cm] (opt-k2) at (blossom-it-k2.east) {\includegraphics[height=2.2cm]{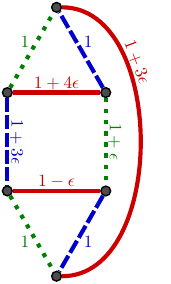}\phantomsubcaption\label{fig:opt-k2}};
\def\RefAnchor{north east}
\node[anchor=\RefAnchor,inner sep=0pt,text height=2ex] at (hyper.north west) {(\subref*{fig:hypercube})};
\path (hyper.north east) -| node[anchor=\RefAnchor,inner sep=0pt,text height=2ex] {(\subref*{fig:opt-k2})}  (opt-k2.north west) ;
\path (hyper.north east) -| node[anchor=\RefAnchor,inner sep=0pt,text height=2ex] {(\subref*{fig:bit-k2})} (blossom-it-k2.north west);
\end{tikzpicture}
\caption{(\subref*{fig:hypercube}): Hypercube $Q_3$ with $3$ ``greedy'' matchings (thick) vs.\ optimum (thin).
$\frac{5}{6}$- and $\frac{7}{9}$-approximations for $\B = 2$ and $\B=3$, respectively, by \BlossomIt{} (\subref*{fig:bit-k2}) vs.\ optimum (\subref*{fig:opt-k2}).
Matchings have the same color/style, the first is red and thick.}%
\FigLabel{greedy-it-approx}
\end{figure}

\GreedyIt{}.
A matching that has at least half the weight of a maximum weight matching can
be obtained by a greedy algorithm in $\bigO(\Sort(m)) \subseteq \bigO(m \log
n)$ time:
Starting from an empty matching, it repeatedly adds the heaviest non-matching
edge $e$ and removes all edges incident to one of $e$'s end vertices until the
graph is empty.
To obtain $\B$ disjoint matchings based on this greedy strategy, it suffices to
sort the edges according to their weight once and construct
the disjoint matchings by iterating $\B$ times over the list of sorted edges
and removing an edge from the list as soon as it becomes part of a matching.
The resulting algorithm \GreedyIt{}
has a running time of $\bigO(\Sort(m) + \B m)$.
The greedy matching algorithm achieves an approximation ratio of $\frac{1}{2}$,
which also transfers to \GreedyIt{}: %
\begin{lemma}\LemLabel{greedy-it-approx}
\GreedyIt{} computes a $\frac{1}{2}$-approximation to the \MDMlong{} problem.
This bound is tight.
\end{lemma}
\begin{proof}%
Let $\DjMatching$ be the solution computed by \GreedyIt{}, let $\DjMatching^*$ be
an optimal solution, and consider an edge $\Edge{u}{v}$ in $\DjMatching^*
\setminus \DjMatching$.
By construction, $\DjMatching$ then contains at least $\B$ edges that are
incident to either $u$ or $v$ and all have weight at least $\Demand(u, v)$.
However, each of the edges in $\DjMatching \setminus \DjMatching^*$ can have
prevented at most two edges in $\DjMatching^* \setminus \DjMatching$ from being
picked themselves, one incident to each of its end nodes.
Hence, $\Demand(\DjMatching^*) \leq 2\cdot \Demand(\DjMatching)$.
The tightness follows by \Lemma{half-approx}.
\end{proof}

\GPAIt{}.
The Global Paths Algorithm (\GPA{}) is a \halfapprox{} algorithm for the
weighted matching problem introduced by Maue and
Sanders~\cite{DBLP:conf/wea/MaueS07}.
It is especially of interest here as the authors have shown that it produces
results that are very close to optimal in experiments, especially if combined
with the postprocessing routine \ROMA{} (see also the end of this section). %
\GPA{} grows a set of paths and even-length cycles as follows:
Initially, every vertex forms a path of zero length.
An edge is called \emph{applicable} if it connects two different paths or the
two end nodes of the same odd-length path.
The algorithm iterates over all edges in weight-descending order and
joins or closes paths by applicable edges.
Afterwards,
it computes an \emph{optimal} weighted matching for each path and even-length
cycle via dynamic programming, which takes time linear in the length of the
path or cycle.
The total running time for \GPA{} hence is $\bigO(\Sort(m) + m)$.
To use $\GPA$ as $\MatchAlg$, it suffices again to sort the edges just once and
only run the path growing and dynamic programming steps $\B$ times, which
results in a total running time of $\bigO(\Sort(m) + \B m)$ for this algorithm,
which we refer to as \GPAIt{}.

\BlossomIt{}.
We also evaluate the use of an optimal weighted matching algorithm as subroutine.
\Blossom{} is the famous algorithm developed by Edmonds originally for the
unweighted matching problem~\cite{edmonds_1965_cardinality}, which he later extended
also to the weighted case~\cite{edmonds1965maximum}.
The key idea is to grow alternating-path trees and shrink odd-length cycles
(called blossoms) to find augmenting paths, which is guided by a (dual) vertex
labelling in the weighted case.
Our algorithm \BlossomIt{} follows the scheme described above and simply
repeats this algorithm $\B$ times, which results in a running time of $\bigO(\B
\cdot n(m + n \log n))$.
\BlossomIt{} does not compute an optimal solution to \MDMlong{}:
As shown in \Figures{bit-k2}{opt-k2},
its approximation ratio can be at most $\frac{5}{6}$ for $\B = 2$ and
$\frac{7}{9}$ for $\B = 3$, \emph{also if all edge weights are set to $1$}.
The solution can be forced by simply setting the weights of the edges in the
first matching to $1 + \epsilon{}$ for some small $\epsilon > 0$ or
to the weights shown in the figure.
Observe that with the shown weights, \GreedyIt{} computes an optimal solution
for both $\B = 2$ and $\B = 3$, \ie, \emph{\BlossomIt{} is not guaranteed to
perform better than \GreedyIt{}.}
However, \BlossomIt{} computes an optimal solution for $\B=1$ and hence
trivially computes a $\frac{1}{2}$-approximation for $\B =2$, as the
optimum for $\B=2$ can be at most twice as large as the optimum for $\B=1$.

\subsection{Algorithms based on Weighted $b$-Matching and Coloring}
\BGreedyExt{}.
We make use of the fact that every \B{}-disjoint matching is also a
$\B$-matching, see also \Section{related}.
Furthermore, the edges of every graph with maximum vertex degree $\Delta$ can
be partitioned into a set of at most $\Delta+1$ matchings by coloring its edges,
which implies that every $(\B{}-1)$-matching can be
translated into a \B{}-disjoint matching without loss of weight.

Analogously to the greedy weighted matching algorithm used in \GreedyIt{},
there is a na\"\i{}ve greedy $b$-matching algorithm that yields a
$\frac{1}{2}$-approximation~\cite{DBLP:conf/esa/Mestre06}.
It iterates over all edges in weight-decreasing order and adds each edge
$\Edge{u}{v}$ to the $b$-matching unless $u$ and $v$ are already incident to
$b$ matching edges.
The running time of this algorithm is $\bigO(\Sort(m))$.
\BGreedyExt{} first obtains a $(\B-1)$-matching and then colors the subgraph
induced by the edges of the $(\B-1)$-matching with the edge coloring algorithm
by Misra and Gries~\cite{DBLP:journals/ipl/MisraG92}, which needs at most $\B$
colors and runs in $\bigO(\B n^2)$ time.
Note that the induced subgraph has a maximum vertex degree of $\B-1$, so the
number of edges is in $\bigO(\B n)$ and the algorithm uses at most $\B$ colors.
The coloring assigns each edge of the subgraph to one of the $\B$ disjoint matchings.
\BGreedyExt{} then runs \GreedyIt{} to enlarge the $\B$ disjoint matchings if possible.
The running time of this algorithm hence is $\bigO(\Sort(m) + \B n^2)$.

All previously described procedures are based on matching or $b$-matching
algorithms, which do not tackle the problem directly and on the whole, but have
a more or less limited view.
We therefore complement our set of algorithms by two further approaches that
try to find a heavy-weight subgraph with chromatic index $\B$.

\nodecentered{}.
The algorithm %
follows a greedy, node-centered strategy:
In a preprocessing step, it calculates a \emph{rating} for each vertex and sorts
the vertices according to their rating.
We consider different options to obtain a vertex's rating from the weights of
its incident edges: the arithmetic mean, the median, the sum, the
maximum, as well as the sum of the $\B$ largest weights (called \NCbsum{}).
Next, it processes the vertices in rating-decreasing order
and tries to color its incident edges in weight-decreasing order.
Each color represents one of the $\B$ disjoint matchings
and the algorithm has to ensure that no vertex is incident to two edges
of the same color.
Hence, if for an edge \Edge{u}{v} the vertices $u$ and $v$ do not share any common
free color, the edge is not picked.
The algorithm stores for each vertex and color a Boolean flag whether this color
has already been used for an incident edge, such that finding a common free
color takes $\bigO(\B)$ time if both end vertices have been matched at most
$\B-1$ times, and $\bigO(1)$ otherwise.

To avoid an overly greedy coloring, we introduce a \emph{threshold}
$\NCthresh \in \Range{0}{1}$ and ignore all edges with weight less than
$\NCthresh\cdot \MaxWeight$, where $\MaxWeight$ is the maximum weight of any
edge.
In this case, the first phase, in which all vertices are processed as described
above, is followed by a second phase, where we merge the sorted lists of the
non-matching edges at each vertex into one sorted list and match and color
greedily.
The running time of the algorithm is $\bigO(\Sort(n) + n \cdot
\Sort(\MaxDegree) + \B m)$ and independent of $\NCthresh$.
Note that $\NCthresh = 0$ is equivalent to setting no threshold.
Irrespective of how the rating function and threshold is chosen, the algorithm
computes the same greedy weighted matching as \GreedyIt{} for $\B = 1$.

\kEC{}.
We also designed a $\B$-edge coloring algorithm that uses the algorithm by
Misra and Gries (\misragries)~\cite{DBLP:journals/ipl/MisraG92} directly as
basis, but only colors up to $\B$ incident edges of each vertex and
takes edge weights into account.
A key property of the \misragries{} algorithm with respect to our modification
is that once an edge has been colored, it may only be recolored later, but
never uncolored.
To color an edge $\Edge{u}{v}$, \misragries{} builds a maximal \emph{fan}
around $u$, which is a sequence $(w_0 = v, w_1, \dots, w_{\ell})$ of distinct
neighbors of $u$ such that, for all $1 \leq i \leq \ell$, if the edge
$\Edge{u}{w_i}$ has color $c$, then $w_{i-1}$ is not incident to an edge with
color $c$, \ie, $c$ is \emph{free} on $w_{i-1}$.
\misragries{} then determines a color $c$ that is free on $u$ and a color $d$
that is free on $w_{\ell}$.
If $d$ is not free on $u$, it looks for a path that starts at $u$ and whose
edges are alternatingly colored $d$ and $c$, and swaps these colors on
the path.
Afterwards, $d$ is guaranteed to be free on $u$.
The prefix of the fan up to the first neighbor $w_j$ where $d$ is free is then
\emph{rotated}, which means that each edge $\Edge{u}{w_i}$ is recolored
with the color of $\Edge{u}{w_{i+1}}$, for all $0 \leq i < j$,
and $\Edge{u}{w_j}$ is colored with $d$.

Our adaptation \kEC{} proceeds as follows:
It processes the edges in weight-descending order and, similar to
\misragries{}, tries to color each edge $\Edge{u}{v}$, however only with one of
up to $k$ colors.
The edge is skipped if $u$ or $v$ are already incident to $k$ colored edges.
If the last neighbor in the maximal fan around $u$ does not have a free color,
\kEC{} tries to color the edge with swapped rules for $u$ and $v$ instead,
and skips the edge if also this fails.
We consider four flags that modify this basic routine:
If \kECcc{} (\emph{common color}) is enabled, \kEC{} tries to find a common
free color of $u$ and $v$ first when trying to color $\Edge{u}{v}$.
With \kEClc{} (\emph{lightest color}), it tries to balance the total weight of
the edges of each color by always picking a free color with minimum total
weight so far.
If \kECrl{} (\emph{rotate long}) is set and the color $d$ is free on $u$, it
rotates the entire fan instead of determining the first neighbor $w_j$ where $d$
is free.
We also consider an option \kEClf{} (\emph{large fan}), where we try to avoid 
neighbors without a free color as long as possible while constructing the fan.
Note that by definition of the fan, a neighbor without free color cannot have a
successor.

As there are at most $kn$ edges that can be colored and coloring
an edge can be done in $\bigO(n)$,
the running time of the algorithm is $\bigO(\Sort(m) + kn^2)$.

\begin{lemma}\LemLabel{half-approx}
The approximation ratio achieved by \GreedyIt{}, \GPAIt{}, \BGreedyExt{}, \nodecentered{}, and \kEC{}
is at most $\frac{1}{2}$.
\end{lemma}
\begin{proof}
Let $\B \geq 1$ and consider the
$\B$-dimensional hypercube graph $Q_\B$, where each vertex additionally has
$\B$ vertices of degree one attached to it, and set all edge weights equal.
A solution returned by one of the mentioned algorithms may consist of all $2^{\B-1}\B$ hypercube
edges, whereas the optimal solution consists of all $2^{\B}\cdot \B$ edges
incident to the vertices of degree one, see also \Figure{hypercube}.
This solution can also be forced by assigning a (very large) weight
$w$ to all non-hypercube edges and an only slightly larger one $w + \epsilon$
to the hypercube edges.
The approximation ratio is then $\frac{w+\epsilon}{2w} = \frac{1}{2} +
\frac{\epsilon}{2w}$.
\end{proof}

\subsection{Postprocessing}\SectLabel{postprocessing}
We consider different postprocessing techniques to improve our algorithms:
\ROMA{} (\emph{Random Order Matching Augmentation}) was
originally proposed for \GPA{} by Maue and Sanders~\cite{DBLP:conf/wea/MaueS07}
to improve the weight of a matching $\Matching$.
For a configurable number $\lROMA{}$ of times, it randomly iterates through
all vertices and for each matched vertex $u$ attempts to improve $\Matching$ by
a so-called maximum-gain \twoaug{}:
A matching edge $\Edge{u}{v} \in \Matching$ is replaced by two non-matching
edges $\Edge{u}{r}$ and $\Edge{v}{s}$
if $r$ and $s$ are currently unmatched and
the gain $\Demand(u, r) + \Demand(v, s) - \Demand(u, v)$ is positive and as large as possible.
The procedure can be terminated early
as soon as one iteration leads to no change in the matching.
We adapt this approach straightforwardly for the $k$-\MDM{} problem by calling
the procedure for each of the $\B$ disjoint matchings separately right after
they have been obtained.
We also consider a variation \Swaps{}, where we instead iterate once
over all matching edges in weight-decreasing order and perform the same
maximum-gain \twoaug{} as in \ROMA{}.

For \Swaps{} as well as each iteration of \ROMA{}, we iterate over all matching
edges and for each of its endpoints explore all neighbors to find the heaviest
incident free edge.
Each matching edge is hence considered once and each non-matching edge at most
twice, which yields a running time complexity of $\bigO(m)$
in case of \Swaps{} under the assumption that the edges are already sorted.
\section{Experiments}\SectLabel{experiments}
We performed extensive experiments to evaluate the performance of
the algorithms described in \Section{algorithms} both with respect to
solution quality and running time.
To keep large numbers readable, we use $\Kilo{}$ and $\Mega{}$ as abbreviations
for thousands ($\times 10^3$) and millions ($\times 10^6$), respectively.
\subsection{Instances, Setup, and Methodology}
The first three collections we use originate from the application side
and contain real-world as well as synthetic instances,
the other two
have been used in previous
experimental evaluations for
$b$-matchings~\cite{DBLP:journals/siamsc/KhanPPSSMHD16}.
\textbf{Facebook Data Traces}~\cite{Roy2015InsideTS} are sets of
production-level traces from three different clusters in Facebook's Altoona Data Center.
For each cluster, there are over
\num{300}\Mega{}
traces, collected over a \SI{24}{\hour} period.
The resulting six instances have
\num{14}\Kilo{} to \num{5}\Mega{}
vertices and
\num{497}\Kilo{} to \num{164}\Mega{}
edges with demands between \num{28} and
\num{262}\Mega{}.
\textbf{HPC} represents four different applications run in parallel using
MPI~\cite{sigmetrics20complexity}.
The instances have
\num{1}\Kilo{}
vertices and up to
\num{38}\Kilo{}
edges;
the demands are between \num{2} and
\num{1}\Kilo{}.
We also use three synthetic \textbf{pFabric} traces~\cite{sigmetrics20complexity,Alizadeh2013pFabricMN},
which have \num{144} vertices and are generated based on flows arriving
according to a Poisson process, with flow rates in $\Set{0.1,0.5,0.8}$.
This results in about
\num{10}\Kilo{}
edges and demands between \num{1} and
\num{34}\Kilo{}.
Following the methodology of Khan
\etal{}~\cite{DBLP:journals/siamsc/KhanPPSSMHD16}, we include
nine instances from the \textbf{Florida Sparse Matrix
Collection}~\cite{Timothy2011_SparseMatrixCollection}.
These instances stem from collaboration networks, %
medical science, %
news networks, %
as well as sensor data, %
electro magnetics, %
and structural mechanics. %
They have
\num{13}\Kilo{} to \num{1}\Mega{}
vertices and
\num{121}\Kilo{} to \num{24}\Mega{} edges,
with demands between \num{1} and
\num{2147}\Mega{}.
Following~\cite{DBLP:journals/siamsc/KhanPPSSMHD16},
we also generated \num{66} Kronecker instances using the Graph500 \textbf{RMAT} generator
with $2^{x}$ vertices, $10 \leq x \leq 20$, and initiator matrices %
\emph{rmat\_b} with $(0.55, 0.15, 0.15, 0.15)$,
\emph{rmat\_g} with $(0.45, 0.15, 0.15, 0.25)$,
and \emph{rmat\_er} with $(0.25, 0.25, 0.25, 0.25)$.
Demands are chosen according to a uniform (\texttt{\_uni})
or exponential distribution (\texttt{\_exp})
and range between \num{1} and half a million.

We implemented %
our algorithms in \texttt{C++17} %
and compiled using \texttt{GCC} 7.5 and full optimization%
 (\texttt{-O3 -march=native -mtune=native}).
For \BlossomIt{}, we adapted an implementation
of the \Blossom{} algorithm from the \texttt{Lemon}%
\footnote{\url{https://lemon.cs.elte.hu/trac/lemon}}
library and made it
iterative by calling the algorithm $\B{}$ times, after each round setting
weights of matched edges to $0$.
We only report results for the variant where \Blossom{} starts with an
approximate matching obtained from a fractional solution instead of an empty
matching, as both compute optimal weighted matchings, but the
fractional option was considerably faster, requiring only half the running time
in the geometric mean across different data sets.

To determine the optimal weight and use it in quality comparisons, we also
implemented an exact algorithm by casting the problem as an ILP,
using an adaptation of the assignment formulation for the edge coloring
problem~\cite{DBLP:conf/latin/JabrayilovM18}, %
and solved it with \texttt{Gurobi}%
\footnote{\url{http://www.gurobi.com}}.
Unfortunately, it only terminated within the timeout of \SI{4}{\hour} for small instances and values of $\B$.

All experiments were performed on a machine with NUMA architecture running
Ubuntu 18.04
with %
\texttt{Intel(R) Xeon(R)} %
CPUs clocked at
\SI{3.40}{\giga\hertz} %
and
\SI{256}{\kilo\Byte} and \SI{20}{\mega\Byte} of L2 and L3 cache, respectively.
The execution of each experiment was pinned to a single CPU and its local
memory to prevent the cost of non-local memory accesses or swapping.
To counteract artifacts of measurement in running time, we ran each experiment
three times and use the median of the \textit{elapsed real time} (wall time).
The only exception to this rule is the ILP, which was run just once as we were
mainly interested in the solution size.
We set a timeout of \SI{4}{\hour}. %

\begin{figure*}[tb]
\def\TextHeight{3.6cm}
\centering
\begin{tikzpicture}[inner sep=0pt]
\matrix (matrix) [matrix of nodes,nodes={inner xsep=1mm,text height=\TextHeight,text width=.3\linewidth,align=left},inner sep=0mm] {%
\includegraphics[height=\TextHeight,width=\linewidth,keepaspectratio]{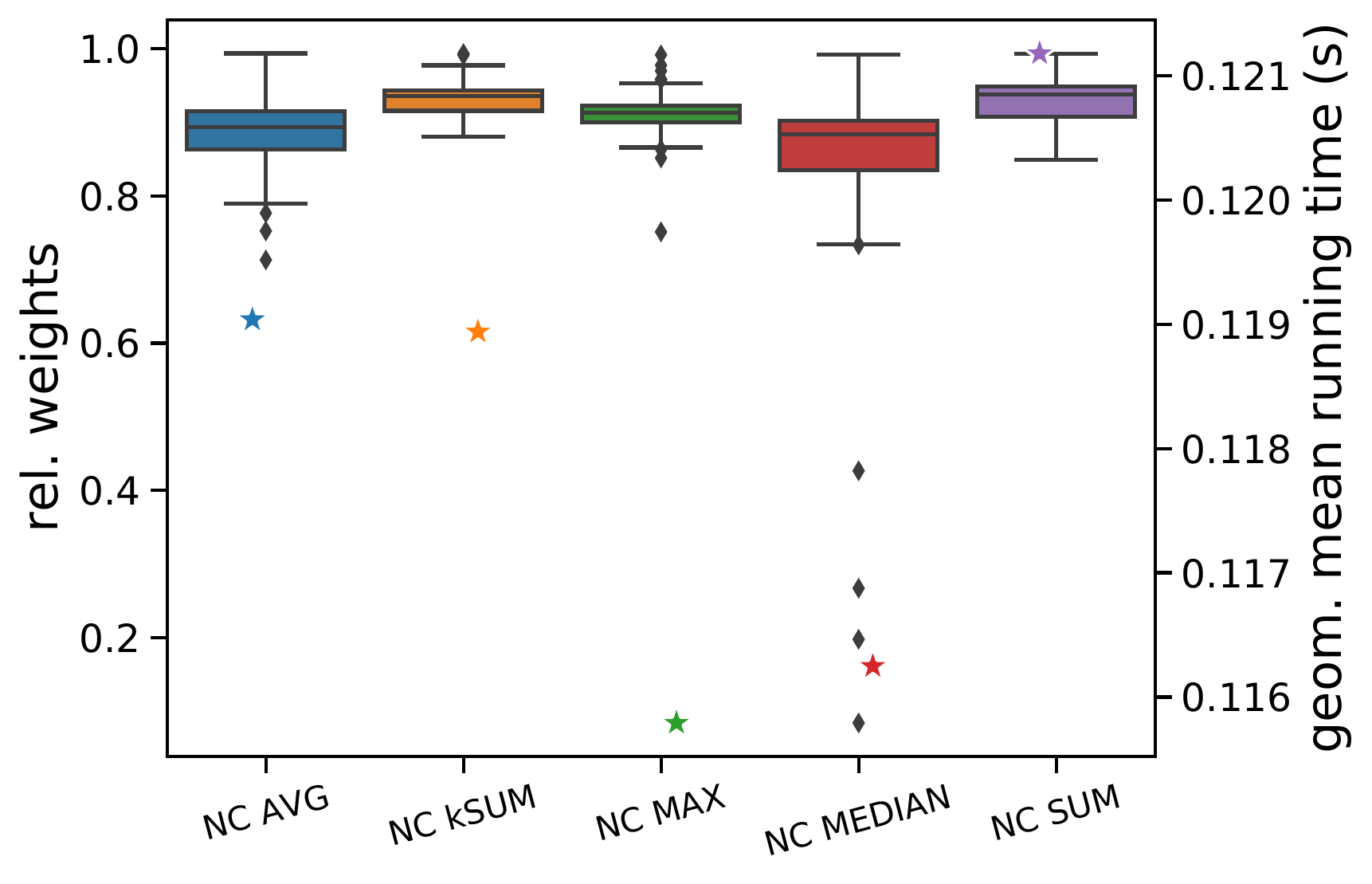}%
\phantomsubcaption\label{fig:nc-all-boxplot-b4}
&
\includegraphics[height=\TextHeight,width=\linewidth,keepaspectratio]{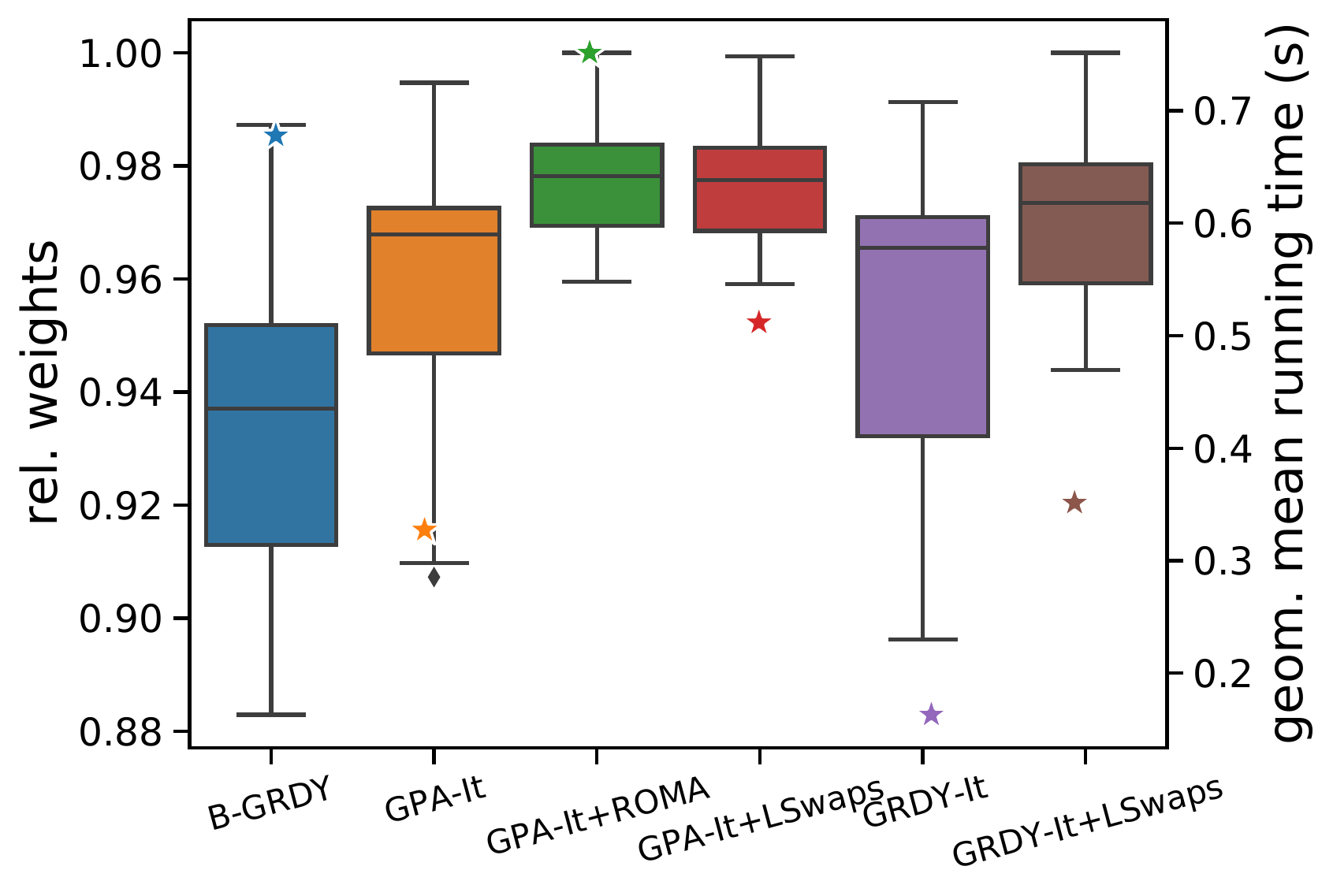}%
\phantomsubcaption\label{fig:greedy-gpa-all-boxplot-b4}
&
\includegraphics[height=\TextHeight,width=\linewidth,keepaspectratio]{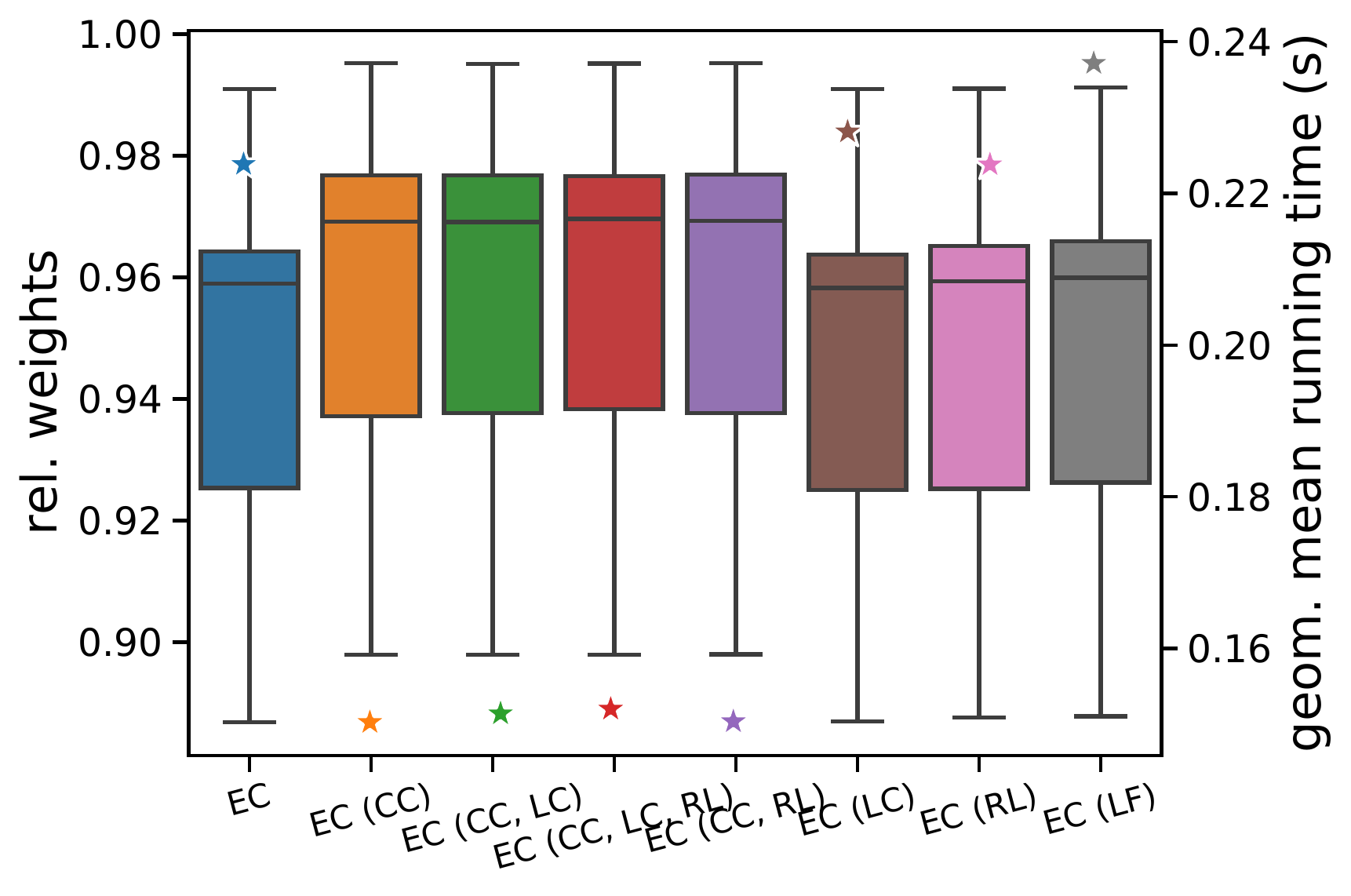}%
\phantomsubcaption\label{fig:kec-all-boxplot-b4}
\\
\includegraphics[height=\TextHeight,width=\linewidth,keepaspectratio]{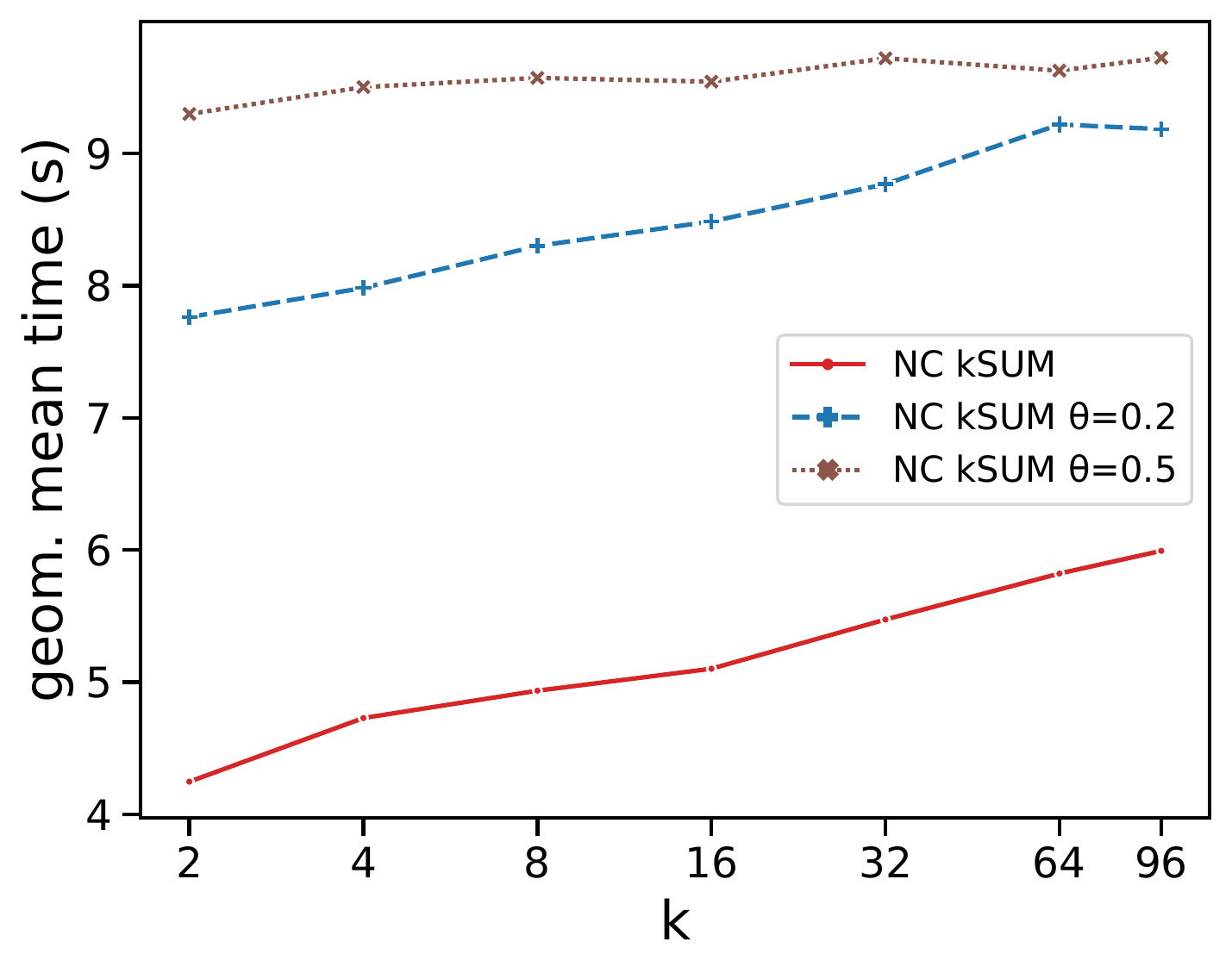}%
\phantomsubcaption\label{fig:best-nc-facebook-timeplot}
&
\includegraphics[height=\TextHeight,width=\linewidth,keepaspectratio]{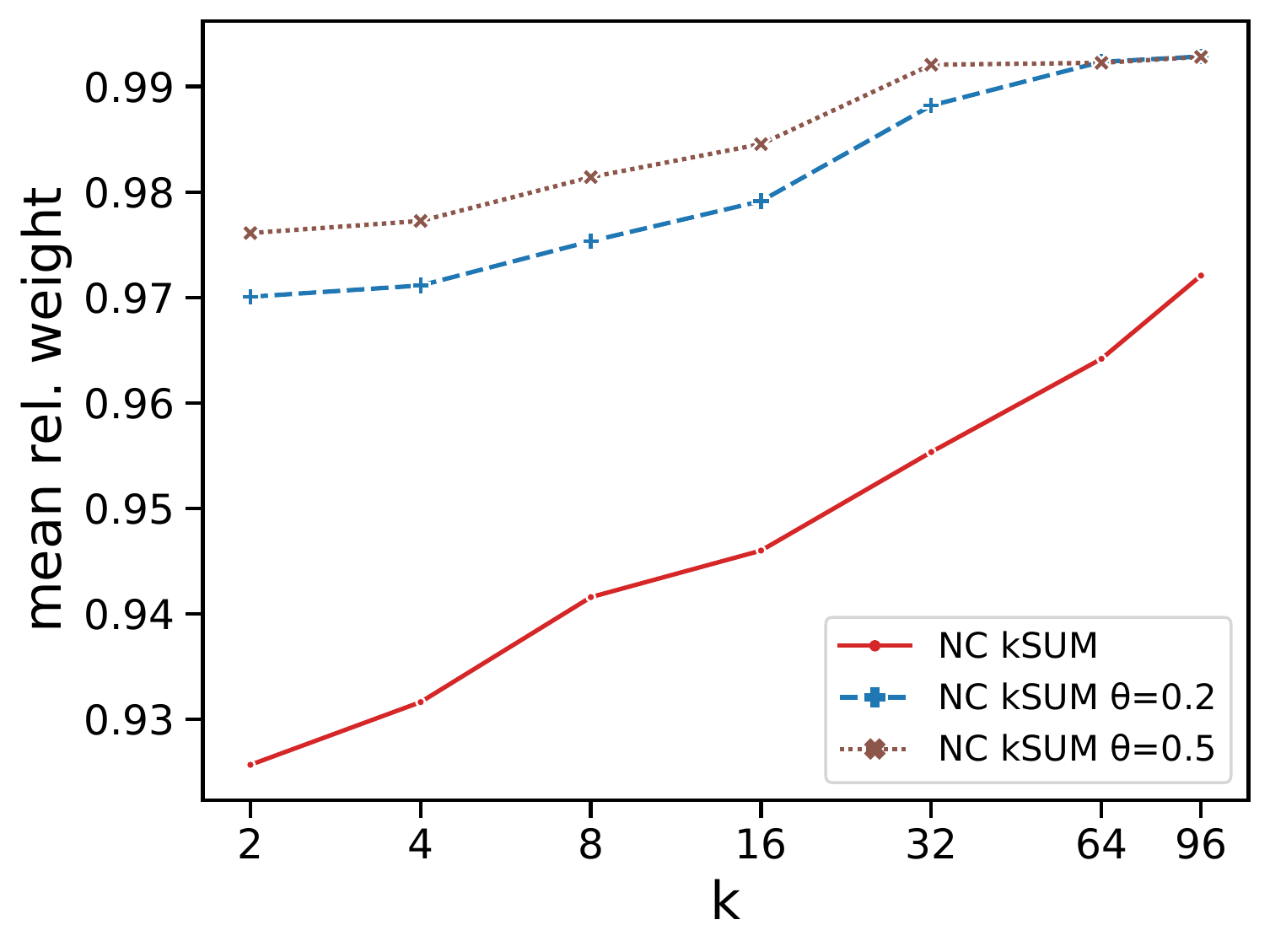}%
\phantomsubcaption\label{fig:best-nc-facebook-weightplot}
&
\includegraphics[height=\TextHeight,width=\linewidth,keepaspectratio]{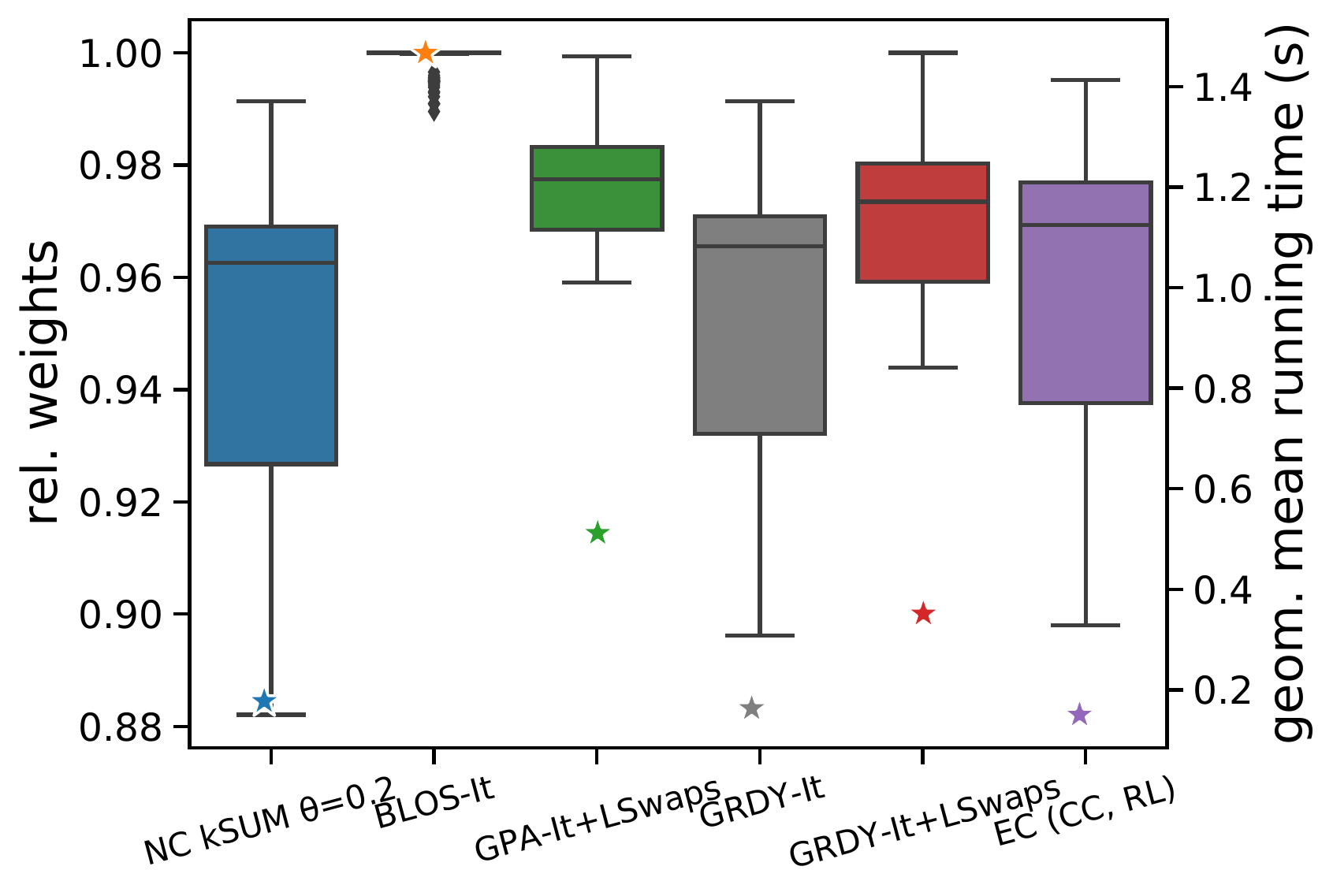}%
\phantomsubcaption\label{fig:best-all-boxplot-b4}
\\
};
\def\RefAnchor{north west}
\node[anchor=\RefAnchor,inner sep=0pt,text height=2ex] at (matrix-1-1.north west) {(\subref*{fig:nc-all-boxplot-b4})};
\node[anchor=\RefAnchor,inner sep=0pt,text height=2ex] at (matrix-1-2.north west) {(\subref*{fig:greedy-gpa-all-boxplot-b4})};
\node[anchor=\RefAnchor,inner sep=0pt,text height=2ex] at (matrix-1-3.north west) {(\subref*{fig:kec-all-boxplot-b4})};
\node[anchor=\RefAnchor,inner sep=0pt,text height=2ex] at (matrix-2-1.north west) {(\subref*{fig:best-nc-facebook-timeplot})};
\node[anchor=\RefAnchor,inner sep=0pt,text height=2ex] at (matrix-2-2.north west) {(\subref*{fig:best-nc-facebook-weightplot})};
\node[anchor=\RefAnchor,inner sep=0pt,text height=2ex] at (matrix-2-3.north west) {(\subref*{fig:best-all-boxplot-b4})};
\end{tikzpicture}
\caption{%
Result quality (left axis) and running time (right axis, depicted as star) for
\nodecentered{} with $\NCthresh = 0$ and different aggregation functions
(\subref{fig:nc-all-boxplot-b4}),
\BGreedyExt{}, \GPAIt{}, and \GreedyIt{} with and without postprocessing
(\subref{fig:greedy-gpa-all-boxplot-b4}),
\kEC{} (\subref{fig:kec-all-boxplot-b4}),
and for the set of the best algorithms (\subref{fig:best-all-boxplot-b4}),
in each case for $\B=4$ and all instance sets.
Running time (\subref{fig:best-nc-facebook-timeplot})
and result quality (\subref{fig:best-nc-facebook-weightplot})
for \nodecentered{} with \NCbsum{} and different thresholds
on \Instance{Facebook}.}%
\FigLabel{algorithm-groups}
\end{figure*}
\subsection{Experimental Results}
We performed experiments for $\B{} \in \Set{2,4,8,16,32,64,96}$.
Our set of algorithms contained
(1)~\nodecentered{} in \num{15} configurations: with threshold $\NCthresh \in \Set{0,0.2,0.5}$ and vertex-ratings
\NCmax{}, \NCavg{}, \NCmed{}, \NCsum{}, and \NCbsum{};
(2)~\GreedyIt{} with and without \Swaps{}; (3)~\GPAIt{} with and without \Swaps{} and additionally with the postprocessing step \ROMA{} ($\lROMA{}=4$) after each iteration,
(4)~\BlossomIt{}, as well as (5)~\BGreedyExt{} and (6)~\kEC{}.
When \Swaps{} were used, they were either applied after each iteration (\LocalSwaps{}) or
once after all iterations finished (\GlobalSwaps{}).

Intuitively, we would expect that it becomes ``easier'' for the algorithms to
add high-demand edges to one of the matchings as $\B$ increases and, thus, that
all algorithms should return an almost equally good solution when $\B = 96$.
This is also confirmed by our results.
Still we can show interesting differences between the algorithms that we will
describe in this section.
We proceed as follows:
We first look at the behavior of similar or the same algorithm with different
configurations, and then compare it to other algorithms using only the best
variant.
Relative solution weights are expressed as a fraction of the optimum (\OPT{})
or, if the optimum is unknown, the best that \emph{any} algorithm has found
(\BEST{}).
Note that all plots use a logarithmic axis for $\B$.

\subsubsection*{\bfseries\nodecentered{}} %
We first consider for each set of instances the relative weights and
mean running times for \nodecentered{} with thresholds $0$, $0.2$, and $0.5$,
respectively, for the five different aggregation functions.
As the threshold effectively limits the number of edges colored in the first
phase and the aggregation function does not play a role in the second phase, we
observe as expected that the differences in quality when using different
aggregation functions become smaller the larger the threshold $\NCthresh$.
In general, \NCmed{} led to worse performance than the other aggregation
methods, especially on the \Instance{Facebook} instance set with no threshold
($\NCthresh = 0$), where it achieved, \eg, for $\B{}=4$ a solution quality of
only \SI{8}{\percent} on \Instance{clusterC-racks}.
This behavior can be explained by the strongly biased demand distribution.
\NCsum{} results in a higher rating of vertices with many (low-demand) edges,
whereas \NCavg{} also takes a vertex's degree into account, which is
however detrimental for small $k$ and a skewed demand distribution.
\NCmax{} can be led astray if vertices have a single edge with very high demand,
but many others with low demand, which resulted in bad performance especially on
the \Instance{pFabric} instances.
Overall, \NCbsum{} showed the \emph{best and most stable performance} in most
cases, see also \Figure{nc-all-boxplot-b4}.

We observe that large thresholds incur a larger time cost, as more edges are left
unprocessed in the first phase and need to be reconsidered in the second
phase.
As an example, \Figure{best-nc-facebook-timeplot} shows the running times
for the \Instance{Facebook}
instances with \NCbsum{}
as aggregation function for the different thresholds;
the behavior on the other instances is similar.
In the worst case, the variants where $\NCthresh > 0$ ran more than twice as
long as without threshold (\SI{125}{\second} vs.\ \SI{260}{\second} (\SI{268}{\second}) for $\B=4$
on \Instance{clusterB-ips} with $\NCthresh=0.2$ ($\NCthresh=0.5$)).
In the geometric mean over all %
instances,
$\NCthresh=0.2$ and $\NCthresh=0.5$
led to a slowdown by a factor between \num{1.5} and \num{2} in comparison to
setting no threshold. %
Looking at the result quality (\Figure{best-nc-facebook-weightplot}), we see
that thresholds are effective in avoiding overly greedy matching, as intended.
The quality differences between $\NCthresh = 0$ and $\NCthresh=0.2$ are ``only''
\SIrange{4}{2}{\percent} %
for \Instance{Facebook} instances, \eg,
which however corresponds to an average absolute gain or loss of \num{1} to
\num{5} billion due to the large absolute values.
Setting $\NCthresh=0.5$ increases the quality only marginally, but comes with
an increased running time, which is why we consider
\emph{\NCbsum{} and
$\NCthresh{}=0.2$ as the best configuration for \nodecentered{}}.

\begin{figure*}[tb]
\def\TextHeight{3.6cm}
\centering
\begin{tikzpicture}[inner sep=0pt]
\matrix (matrix) [matrix of nodes,nodes={inner xsep=3mm,text height=\TextHeight},inner sep=0mm] {%
\includegraphics[height=\TextHeight]{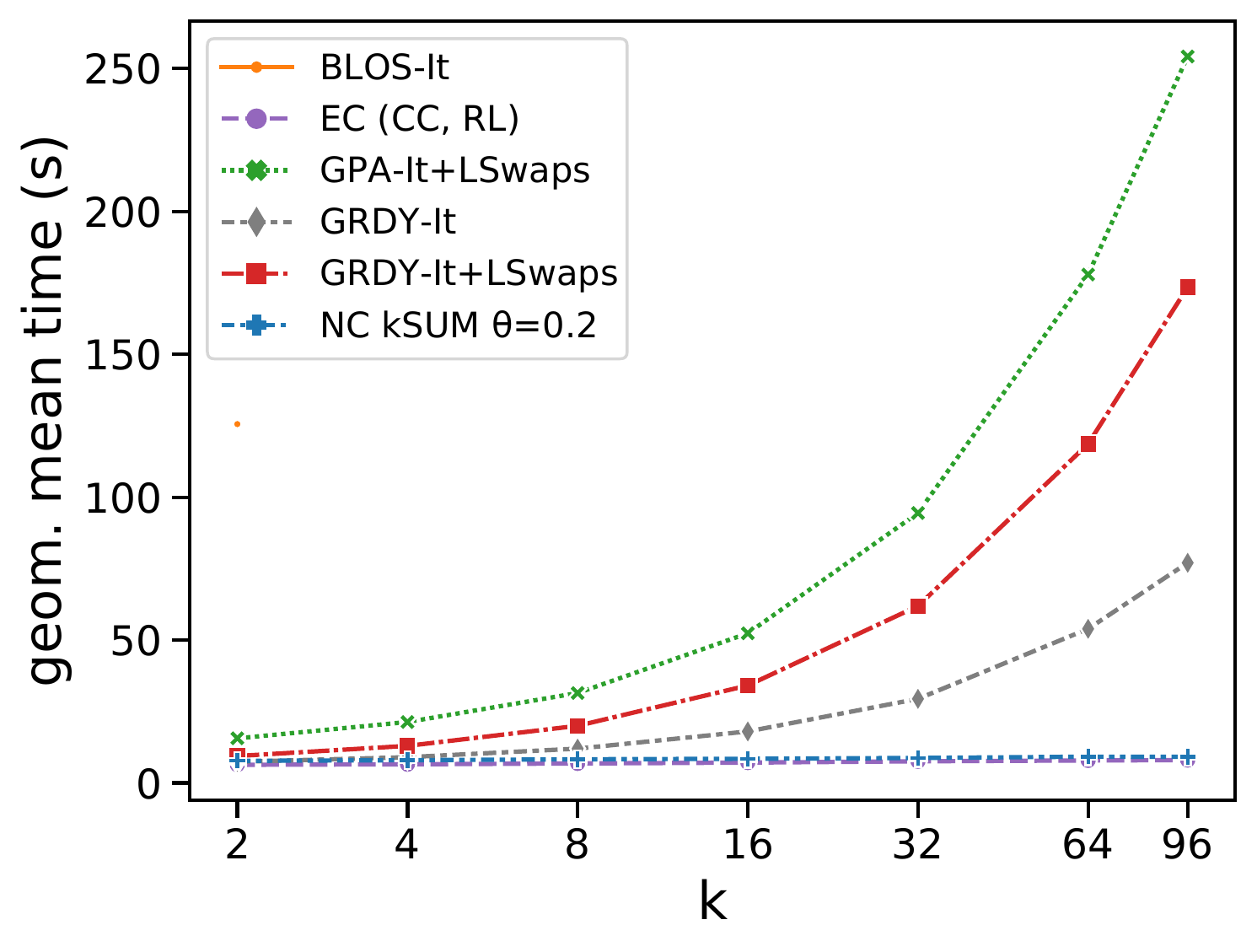}%
\phantomsubcaption\label{fig:best-facebook-timeplot}
&
\includegraphics[height=\TextHeight]{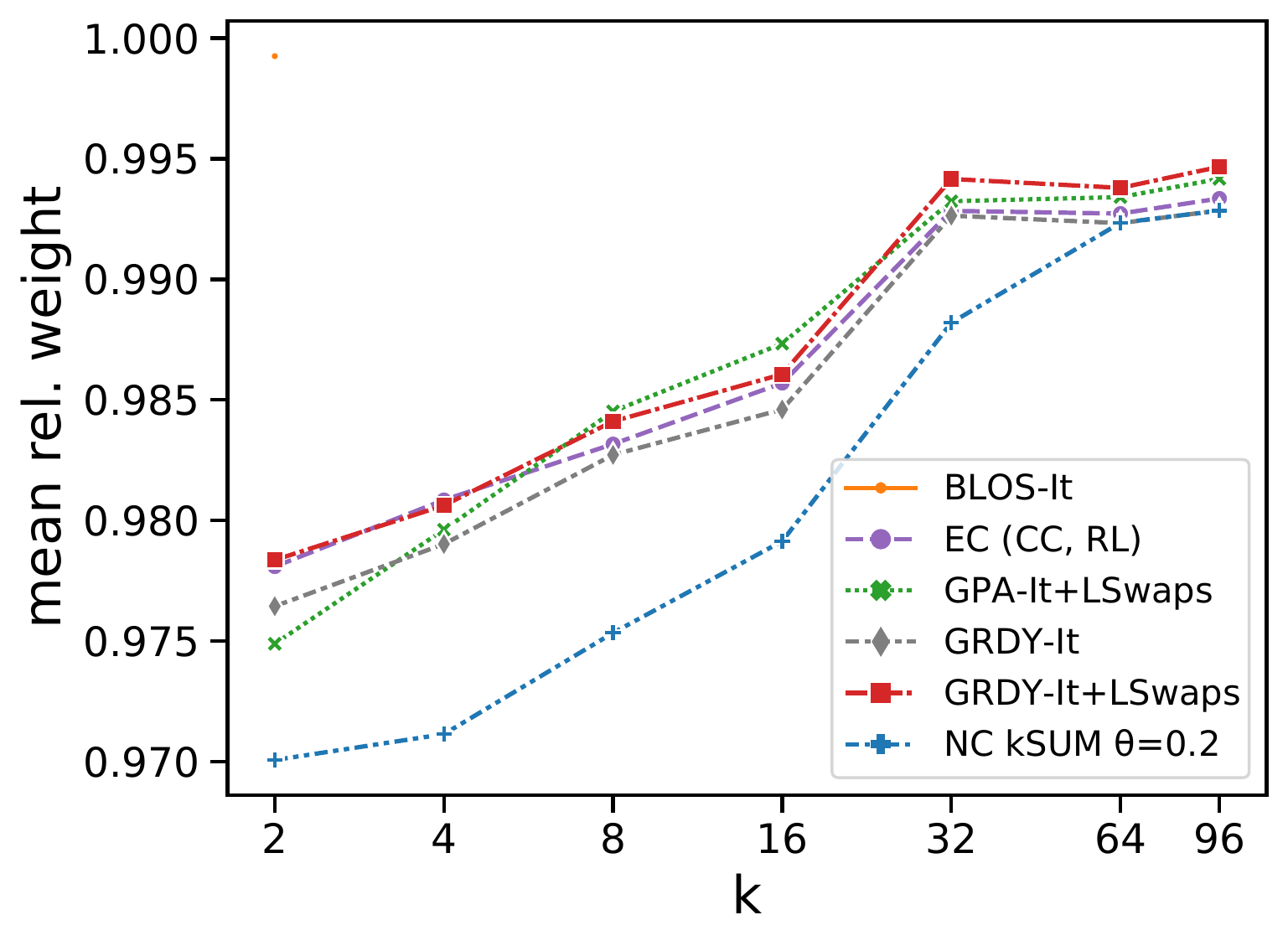}%
\phantomsubcaption\label{fig:best-facebook-weightplot}
&
\includegraphics[height=\TextHeight]{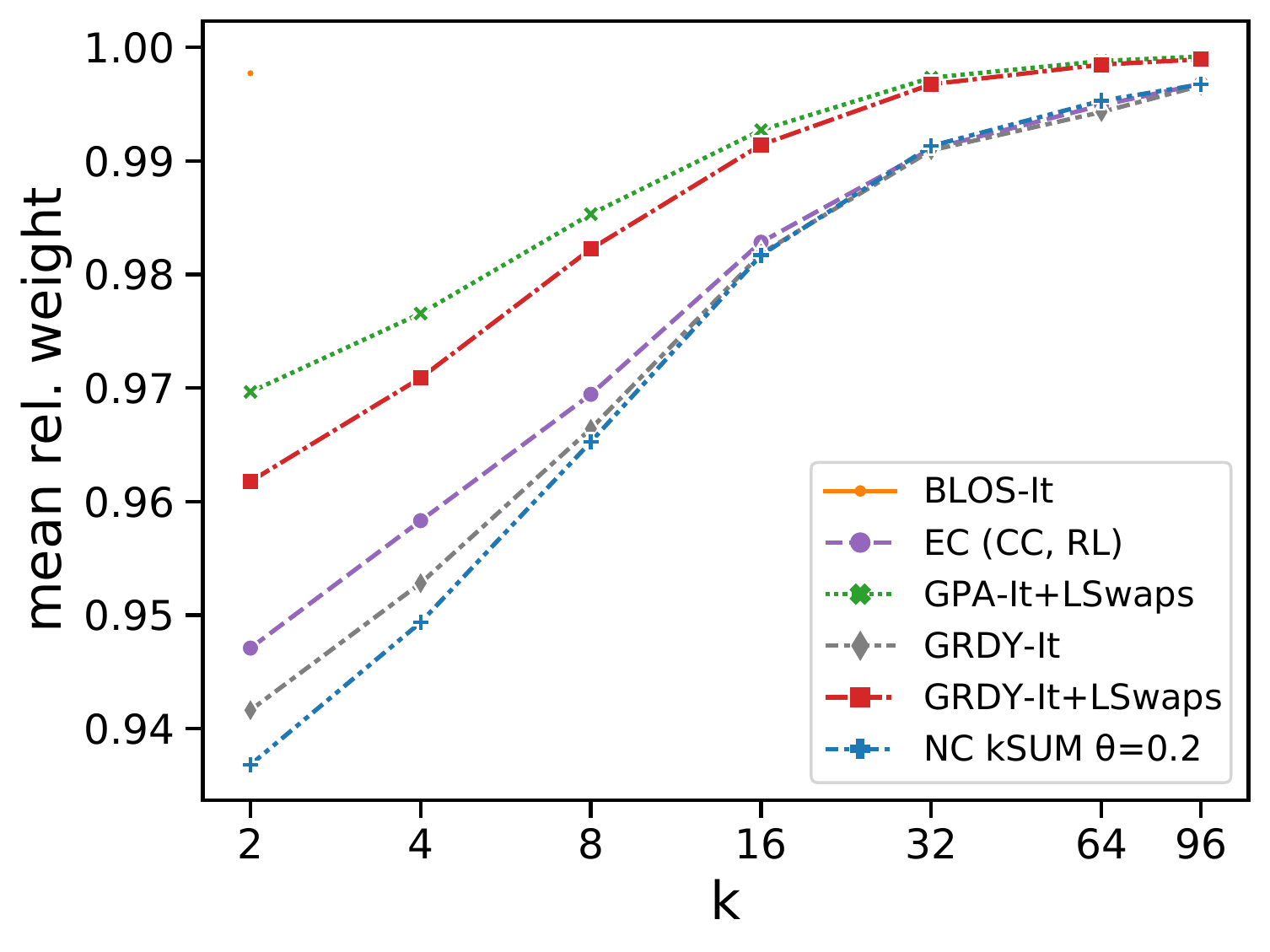}%
\phantomsubcaption\label{fig:best-all-weightplot}
\\
\includegraphics[height=\TextHeight]{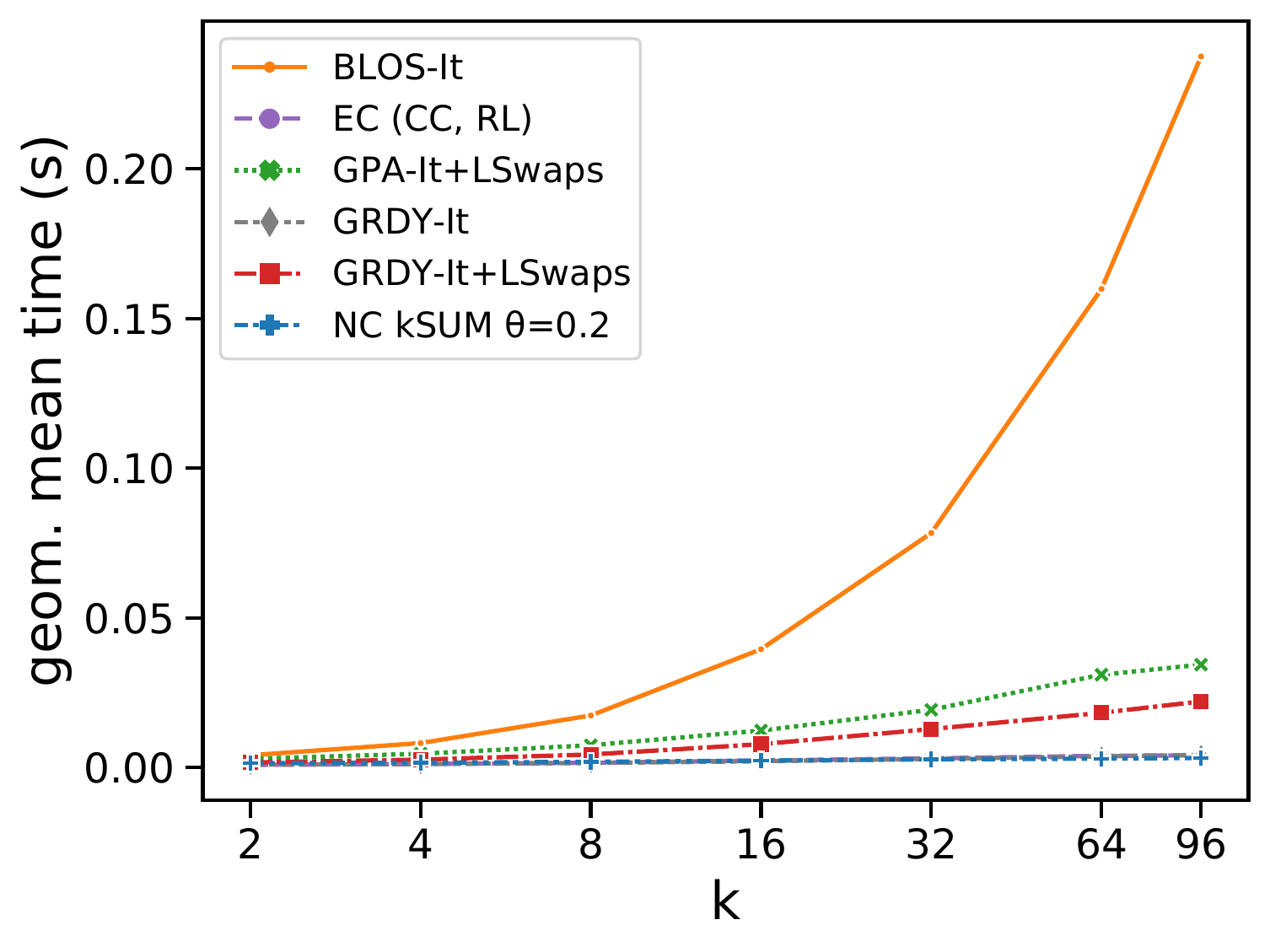}%
\phantomsubcaption\label{fig:best-hpc-pfab-timeplot}
&
\includegraphics[height=\TextHeight]{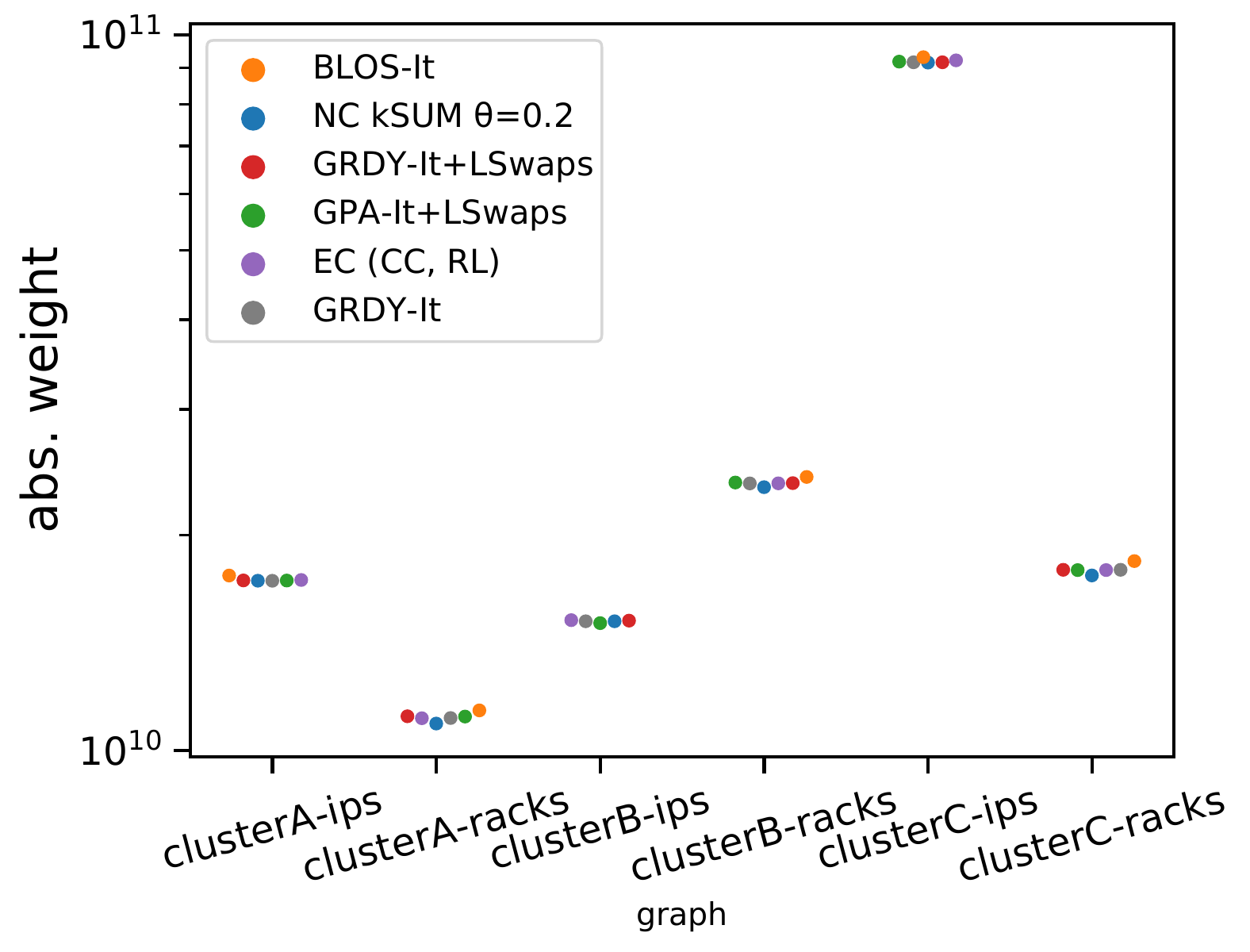}%
\phantomsubcaption\label{fig:best-facebook-instance-weight-b4-abs}
&
\includegraphics[height=\TextHeight]{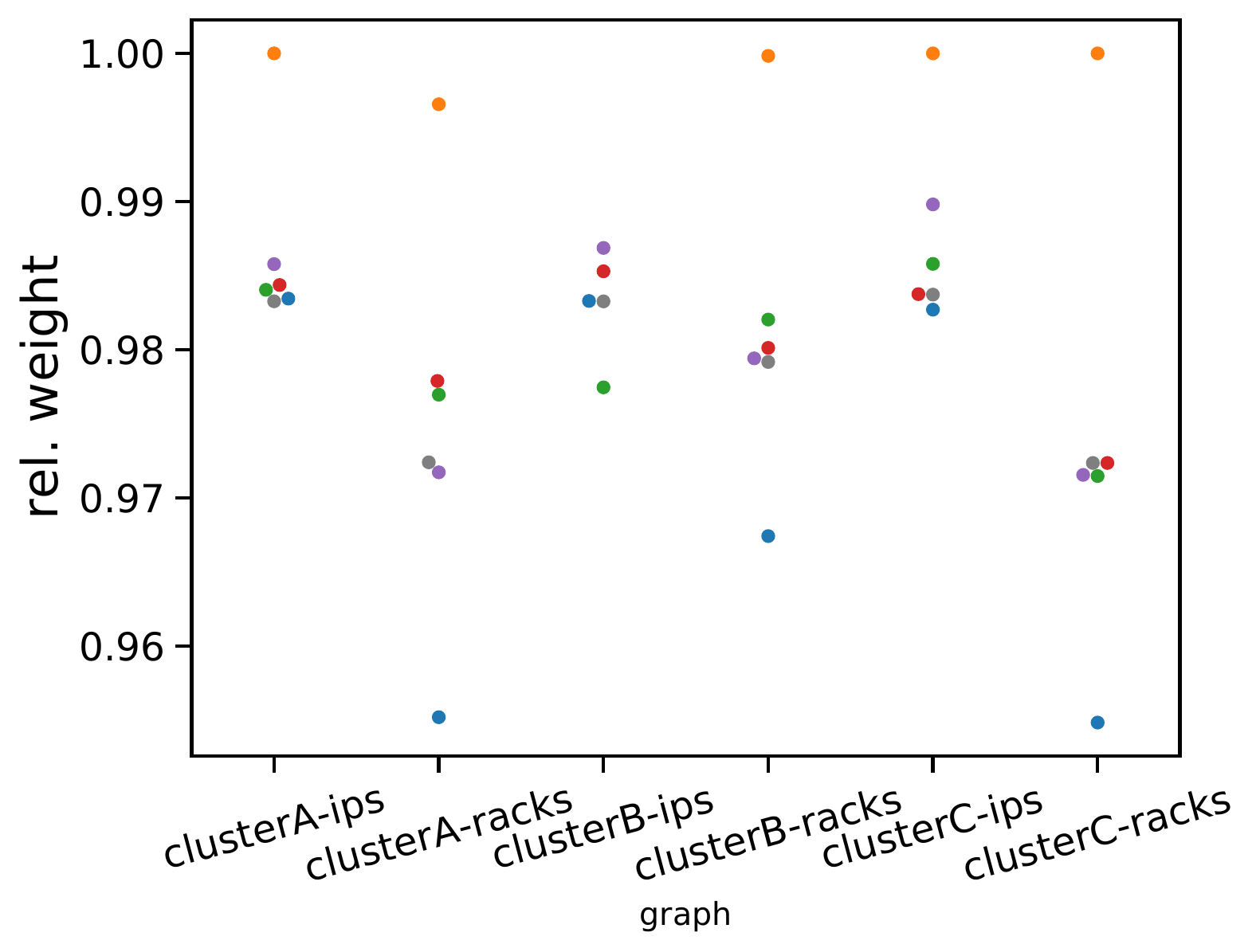}%
\phantomsubcaption\label{fig:best-facebook-instance-weight-b4-rel}
\\
};
\def\RefAnchor{north}
\node[anchor=\RefAnchor,inner sep=0pt,text height=2ex,xshift=1ex] at (matrix-1-1.north west) {(\subref*{fig:best-facebook-timeplot})};
\node[anchor=\RefAnchor,inner sep=0pt,text height=2ex,xshift=1ex] at (matrix-1-2.north west) {(\subref*{fig:best-facebook-weightplot})};
\node[anchor=\RefAnchor,inner sep=0pt,text height=2ex,xshift=1ex] at (matrix-1-3.north west) {(\subref*{fig:best-all-weightplot})};
\node[anchor=\RefAnchor,inner sep=0pt,text height=2ex,xshift=1ex] at (matrix-2-1.north west) {(\subref*{fig:best-hpc-pfab-timeplot})};
\node[anchor=\RefAnchor,inner sep=0pt,text height=2ex,xshift=1ex] at (matrix-2-2.north west) {(\subref*{fig:best-facebook-instance-weight-b4-abs})};
\node[anchor=\RefAnchor,inner sep=0pt,text height=2ex,xshift=1ex] at (matrix-2-3.north west) {(\subref*{fig:best-facebook-instance-weight-b4-rel})};
\end{tikzpicture}
\caption{%
Mean running time (\subref{fig:best-facebook-timeplot})
and result quality (\subref{fig:best-facebook-weightplot})
of the best algorithms on \Instance{Facebook},
mean result quality on all \num{88} instances (\subref{fig:best-all-weightplot}),
and
mean running time on \Instance{HPC} and \Instance{pFabric} (\subref{fig:best-hpc-pfab-timeplot}).
(\subref{fig:best-facebook-instance-weight-b4-abs}, \subref{fig:best-facebook-instance-weight-b4-rel}):
Per-instance absolute and relative weights %
on \Instance{Facebook}.
}%
\FigLabel{best-comparison}
\end{figure*}

\subsubsection*{\bfseries \GPAIt{}, \GreedyIt{}, \BGreedyExt{}}
As an example, \Figure{greedy-gpa-all-boxplot-b4} shows the result quality
and running time for $\B=4$.
We can observe a boost in quality for \GPAIt{} when activating either
\LocalSwaps{} or \ROMA{}, but no improvement with \GlobalSwaps{} (omitted in the plot).
Comparing \LocalSwaps{} and \ROMA{}, we obtain almost equal result quality
at distinctly faster speed with \LocalSwaps{}.
The running time with \GlobalSwaps{} is similar to plain \GPAIt{} and faster
than with \LocalSwaps{} by a factor of two with $\B{}=96$.
As \GlobalSwaps{} has almost equal quality as \GPAIt{} this suggests that no or only very few swaps were performed.
Given the trade-off between quality and running time, %
we consider \emph{\GPAIt{} with \LocalSwaps{}} to be the better option, which we
will use in our further analysis.

Similarly, we evaluated \GreedyIt{} with local and global
\Swaps{}, compared to a base version without swaps.
Again, we observe barely any improvement in quality by \GlobalSwaps{}, yet
\LocalSwaps{} consistently yields results with better quality,
at the expense of an increased running time. %
To consider both ends of the result quality vs.\ running time tradeoff, we
include both \GreedyIt{} alone as well as \GreedyIt{} with
\LocalSwaps{} in our further analysis.

\BGreedyExt{} was inferior to \GreedyIt{} on all sets of instances both with
respect to running time and solution quality and is therefore not considered
further.

\subsubsection*{\bfseries \kEC{}}
\Figure{kec-all-boxplot-b4} shows the result quality and running time for \kEC{}
with different combinations of flags for $\B=4$ and all instances.
As expected, \kECcc{} (common color) decreases the running time, here by over
\SI{30}{\percent}, as fan construction and rotation are no longer required in
many cases.
It increases the result quality distinctly as it can also color an edge if
the last neighbor in the fan does not have a free color.
To the contrary, \kEClc{} (lightest
color) leads to a clearly visible decline without \kECcc{} and in general to a
slight increase in running time due to the additional maintenance of color
weights.
\kECrl{} (rotate long) marginally improves result quality and has
a negligible effect on the running time,
whereas \kEClf{} leads to a slowdown in general and slightly better
quality only if \kECcc{} is not set.
We thus consider \kECcc{} and \kECrl{} as the best parameters for
\kEC{}.

\subsection{Overall Running Times and Result Quality}
Given our choice of representatives for each algorithm, we analyze these
representatives regarding their running time and result quality in detail on
the instance set \Instance{Facebook} and
only give a summary about the others.
\emph{We do not discuss the other instance sets in detail any further,
as the algorithms perform very consistently on all of them.
}
Note that a given algorithm is only represented for a given $\B{}$ if that
algorithm finished on \emph{all} instances of a set within our \SI{4}{\hour}
time limit.

Looking purely at the running time complexities (cf.~\Table{complexities}), one
might expect to see \GPAIt{} and \GreedyIt{} behaving similarly to \nodecentered{}.
The former two have a slightly larger preprocessing time,
yet afterwards all perform $\bigO(\B m)$ work (with and without \LocalSwaps{}).
\kEC{}, on the other hand, has both large preprocessing time and performs
$\bigO(\B n^2)$ work, so it could be expected to be the slowest.
However, \Figures{best-facebook-timeplot}{best-hpc-pfab-timeplot} paint a
vastly different picture, as \nodecentered{} and \kEC{} compute the disjoint
matchings significantly faster than the \AlgName{*-It} algorithms.
This can be observed consistently on all instances.
For the \Instance{Facebook} instances and $\B{}=4$, \kEC{} achieves
in the geometric mean
a speedup of \num{2} and \num{3.2} over \GreedyIt{} with \LocalSwaps{}
and \GPAIt{}, respectively, and for $\B{}=96$ the speedups increase to
\num{21.6} and \num{31.7}.
The running time of \GreedyIt{} without \LocalSwaps{} is less than the time for
\GreedyIt{} with \LocalSwaps{}, but larger than for \kEC{}.
\nodecentered{} is equally fast as \kEC{}.
Over all instances and values of $\B$, \kEC{} is faster than \GreedyIt{} and
\GPAIt{} with \LocalSwaps{} by a factor of \num{1.8} to \num{7.3} and \num{2.6}
to \num{9.5}, respectively.
\BlossomIt{} terminated on all instances in \Instance{HPC} and
\Instance{pFabric}, but was \num{4.5} to \num{57.9} times slower than \kEC{}.
The speedup by \kEC{} over the plain variant of \GreedyIt{} is less
pronounced, but still between \num{1.2} and \num{9.6} for \Instance{Facebook} and
up to \num{2}, \eg, on \Instance{Florida}.
The reason that \GreedyIt{} without \LocalSwaps{} is faster than \GreedyIt{}
with \LocalSwaps{} is that \LocalSwaps{} prevents the algorithm from
efficiently cutting down the list of edges to process in the next each
iteration:
As \LocalSwaps{} changes the matching, the non-matching edges need either be
sorted after each iteration or all edges are processed in each iteration,
causing $\Theta(\B)$ work per edge.
\nodecentered{} and \kEC{}, on the other hand, operate more locally.
\nodecentered{} scans each edge at most three times and
\emph{only if} both end vertices have been matched less than $\B$ times so far,
it compares two lists of Boolean arrays of length $\B$ to determine a common
free color.
Thus, the work per edge is often just constant.
The situation is similar for \kEC{}.

Regarding quality
(\FiguresSub{best-facebook-weightplot}{best-all-weightplot}), for
\Instance{Facebook} instances and $\B{} \leq 4$, \kEC{} and \GreedyIt{} with
\LocalSwaps{} perform best.
For $\B \geq 8$, \kEC{} stays slightly behind \GreedyIt{} and \GPAIt{} with \LocalSwaps{}
by less than \SI{0.01}{\percent}
(regarding the mean of weights relative to \BEST{}).
The mean performance of \nodecentered{} always remains within \SI{1}{\percent}
of \BEST{}.
Across \emph{all instances}, \GPAIt{} with \LocalSwaps{} performed best, with a mean
relative weight of at least \SI{97}{\percent} of \BEST{},
closely followed by \GreedyIt{} with \LocalSwaps{} and \kEC{} with at least \SI{96}{\percent}
and \SI{94.7}{\percent} on average.
\nodecentered{} performed worst, however still within \SI{93.7}{\percent}
of \BEST{} on average for $\B=2$ and \SI{99.7}{\percent} for $\B=96$.
If we look at the \emph{worst performance} per algorithm across all instances,
we observe a quality ratio of at least \SI{95}{\percent} for \GPAIt{},
\SI{93}{\percent} for \GreedyIt{} with \LocalSwaps{},
\SI{87}{\percent} (\SI{90}{\percent} for $\B \geq 4$) for \GreedyIt{} and \kEC{},
and \SI{76}{\percent} (\SI{88}{\percent} for $\B \geq 4$) for \nodecentered{},
see also \Figure{best-all-boxplot-b4}.

\Figures{best-facebook-instance-weight-b4-abs}{best-facebook-instance-weight-b4-rel}
show absolute and relative per-instance weight comparisons for $\B{}=4$ on the
\Instance{Facebook} instances.
We can clearly observe that \nodecentered{} struggles with the rack-level instances \Instance{clusterA-racks}, \Instance{clusterB-racks}, and \Instance{clusterC-racks}.
\kEC{} is second-best after \BlossomIt{} on the IP-level instances.

\BlossomIt{} finished within the \SI{4}{\hour} time limit on \SI{97}{\percent} of all experiments.
Its asymptotic running time has an additive factor of $\bigO(\B n m)$ as
compared to $\bigO(\B m)$ for the other algorithms and this is confirmed by our
experiments: It is the slowest on all graphs.
However, it always achieves the best quality results and for all graphs where
the ILP terminated, \BlossomIt{} was within \SI{99}{\percent} of the result quality of
the optimum. %
However, for $\B \ge 64$ the faster algorithms achieved almost the same result
quality. Thus, \BlossomIt{} is a good choice only for small values of $\B$ and
in settings where running time is not a limiting factor.

The ILP completed on all \Instance{HPC} instances for $\B{} \in \Range{2}{16}$,
all \Instance{pFabric} instances for $\B{} \in \Range{2}{8}$,
most \Instance{RMAT} instances with $n=2^{10}\dots 2^{12}$ for $\B \in \Range{2}{16}$
(\Instance{rmat\_er}, \Instance{rmat\_g\_12} only for $\B \in \Set{2,4}$),
as well as on three \Instance{Facebook} instances
(\Instance{clusterA-racks} and \Instance{clusterB-racks} for $\B{} \in \Set{2,
4}$ and \Instance{clusterC-racks} only for $\B = 2$).

The order of the algorithms with respect to running time and result quality is
consistent on all instances except for \Instance{pFabric}, where \kEC{} on
average finds larger solutions than \GreedyIt{} with \LocalSwaps{} and
partially also \GPAIt{} for all values of $\B$.

\emph{Overall we conclude that for medium and large values of $\B$, \kEC{} with
\kECcc{} and \kECrl{} enabled is the best-performing algorithm.}
Unlike the running time of the \AlgName{*-It} algorithms, its running time barely increases with
$\B$ and its quality score is \emph{on average} within \SI{95}{\percent} or more of the best
algorithm within the \SI{4}{\hour} time limit, and \SI{98}{\percent} or better
for $\B \geq 32$.
\emph{It is also one of the best algorithms for small values of $\B$ on the
\Instance{Facebook} and \Instance{pFabric} instances}.
\emph{On the other instances, \GPAIt{}  with \LocalSwaps{} is a good choice for small
values of $\B$,
as its quality is always within \SI{95}{\percent} of the best algorithm while its
running time is still moderate; if running time is not of concern, \BlossomIt{}
is a better choice.}

\section{Future Work}\SectLabel{conclusion}
There remain several interesting avenues for future work.
In particular, it would be interesting to further explore 
the power of randomized algorithms. 
The only randomized algorithm we analyzed is \GPAIt{} with \ROMA{}, but it did not show the strongest performance.
On the practical front, it will be interesting to deploy and experiment with our algorithms in a small 
datacenter network using optical circuit switches.

The authors have provided public access to their code at \url{https://doi.org/10.5281/zenodo.5851268}.

\balance
\bibliography{paper}
\end{document}